\documentclass[12pt]{article}
 
\usepackage{amsmath, amssymb, amsfonts,amsthm} 
\usepackage{graphicx,color}
\usepackage{hyperref}  
\hypersetup{hypertexnames=false}
\usepackage{physics}  
\usepackage{bm} 
\usepackage{geometry}  
\usepackage[numbers]{natbib}
  
\title{Model-Free Quantum Stabilization via Finite-Difference Lyapunov Control}

\author{
	Robert Vrabel\thanks{Corresponding author. Email: robert.vrabel@stuba.sk. 
		This is an Accepted Manuscript of an article published by Taylor \& Francis 
		in \emph{International Journal of Control} on 09 April 2026. 
		The Version of Record is freely available at: 
		\url{https://www.tandfonline.com/doi/full/10.1080/00207179.2026.2656156}.} \\
	Slovak University of Technology in Bratislava, \\
	Institute of Applied Informatics, Automation and Mechatronics, \\
	Bottova 25, 917 24 Trnava, Slovakia
}
\date{}

\theoremstyle{plain}
\newtheorem{theorem}{Theorem}[section]
\newtheorem{lemma}[theorem]{Lemma}
\newtheorem{proposition}[theorem]{Proposition}
\newtheorem{corollary}[theorem]{Corollary}

\theoremstyle{definition} 
\newtheorem{definition}[theorem]{Definition}
\newtheorem{problem}[theorem]{Problem}

\theoremstyle{remark}
\newtheorem{remark}[theorem]{Remark}

\newenvironment{keywords}{
	\vspace{2mm}
	\noindent\small\textbf{Keywords:}\ \itshape}{\par\vspace{2mm}}

\begin{document}

	\maketitle
\begin{abstract}
	We develop a model-free framework for stabilizing quantum states using only empirical finite-difference evaluations of a measurement-derived Lyapunov observable. The controller requires no knowledge of the Hamiltonian, dissipative structure, or generator of the dynamics, and relies solely on discrete measurement data. The approach combines three key elements: sign-based Lyapunov descent, adaptive gain amplification, and a finite-difference analogue of LaSalle’s invariance principle. We provide rigorous conditions under which these mechanisms guarantee asymptotic stabilization along the sampling instants in the drift-free case and practical input-to-state stability (ISS) in the presence of unknown drift and noise.
	 The resulting feedback law is simple, derivative-free, and experimentally feasible. A qubit example illustrates the complete closed-loop scheme and the predicted ISS-type behavior. Although demonstrated on a single qubit, the theory applies to arbitrary finite-dimensional quantum systems and offers a foundation for further developments in stochastic, subspace, and multi-qudit model-free quantum control.
\end{abstract}


\section{Introduction}\label{sec:Introduction}
Stabilization of quantum states is a core requirement in quantum information processing, high-precision sensing, and coherent manipulation of nanoscale systems. Classical feedback theory provides powerful tools for ensuring stability under uncertainty, yet its direct application to quantum systems is severely constrained: the controller typically lacks knowledge of the system Hamiltonian, the dissipative mechanisms, or even the structure of the effective generator. Modern experimental platforms--including superconducting qubits, trapped ions, and photonic architectures--operate in regimes where device parameters drift and environmental interactions cannot be accurately identified~\cite{wiseman_milburn_2010}. A recent survey~\cite{weidner2025robust} highlights that this mismatch creates a persistent gap between experimental practice and the assumptions underlying most existing quantum-control strategies, noting in particular that classical robust-control methods require a level of model knowledge rarely available in quantum settings.

This discrepancy reveals a fundamental tension between model-based stabilization methods and the information structure encountered in actual laboratory conditions. Approaches based on Hamiltonian engineering, measurement-based feedback, or optimal-control design almost invariably rely on at least partial knowledge of the dynamics~\cite{dong_petersen_2010,altafini_ticozzi_2012,ticozzi_viola_2008}. Even in continuous-measurement feedback and quantum filtering~\cite{belavkin_1983,wiseman_milburn_1993,bouten_2007}, the drift and noise operators must be specified or estimated. 
In contrast, real-time experiments often operate directly from streaming
measurement data without reconstructing any dynamical model, creating a gap
between available control-theoretic tools and the information structure
encountered in practice.

The present work establishes the foundations of such a framework. Our approach is built on two key ideas. First, since analytic derivatives of a Lyapunov function are inaccessible when the dynamics are unknown, we rely exclusively on \emph{finite-difference information} extracted from measurement outcomes. This naturally leads to a discrete-time analogue of LaSalle’s invariance principle, formulated in terms of observable differences instead of derivatives. While related concepts appear in derivative-free stability analysis~\cite{khalil_2002}, no adaptation to quantum systems has previously been available. Second, we show that stabilization in the presence of unmodeled drift and noise admits a quantum analogue of classical input-to-state stability (ISS)~\cite{sontag_1989,jiang_mareels_wang_1996}. Unknown Hamiltonian drift may fundamentally prevent exact convergence; nevertheless, practical stability within a disturbance-dependent neighborhood of the target state can be guaranteed.

To support these developments, we introduce several structural notions tailored
to in\-for\-mation-limited quantum control: adaptive Lyapunov observables,
perturbation-based descent control, and model-free stabilizability. These allow
the controller to determine descent directions using only measurement data,
without any form of model identification. The resulting stabilization mechanism
operates entirely without knowledge of the generator and exhibits behavior
analogous to Lyapunov-based feedback in nonlinear control.

Recent years have seen substantial progress in quantum control from a
	control-theoretic perspective. Lyapunov-based stability analysis and
	stabilization techniques have been developed for both closed and open quantum
	systems, including invariance principles and convergence guarantees
	\cite{dong_petersen_2010,altafini_ticozzi_2012,ticozzi_viola_2008,
		ticozzi_viola_2012,Emzir2022}. These works establish a rigorous foundation for
	feedback and open-loop control design, but typically rely on explicit knowledge
	of the system generator or its parametrization.

	In parallel, robust and switching-based Lyapunov control strategies have been
	proposed for open quantum systems affected by decoherence and dissipation.
	Recent results demonstrate practical stability, finite-time convergence, or
	contractive behavior under switching control laws
	\cite{wu_switching_qubit_2025}. While such approaches significantly relax the
	requirement of asymptotic convergence in open quantum systems, they still
	presuppose detailed knowledge of the system dynamics and admissible target
	structures, which limits their applicability in information-limited settings.

	Measurement-based feedback and stochastic control formulations have been studied
	using quantum filtering and continuous measurement models
	\cite{wiseman_milburn_2010,belavkin_1983,wiseman_milburn_1993,bouten_2007}.
	Such approaches provide powerful tools for real-time control under uncertainty,
	yet still require specification or estimation of the drift and noise operators.
	Related stability analyses for open quantum systems have also been developed
	using operator-theoretic and semigroup-based methods
	\cite{lindblad_1976,gorini_kossakowski_sudarshan_1976,pan_amini_miao_gough_ugrinovskii_james_2014}.

	More recently, learning-based approaches have been explored for measurement-based
	quantum feedback control, including reinforcement learning and deep reinforcement
	learning schemes that construct control policies directly from measurement data
	\cite{song_drl_feedback_2025}. These methods demonstrate impressive empirical
	performance, such as accelerated convergence and robustness to delays and
	imperfect measurements. However, they rely on offline training, reward-function
	design, and repeated interactions with simulated or experimental systems, and
	do not provide Lyapunov-based or invariance-type stability guarantees.

	From a broader nonlinear control viewpoint, robustness and disturbance rejection
	are naturally captured by input-to-state stability (ISS) concepts
	\cite{sontag_1989,jiang_mareels_wang_1996}, while classical Lyapunov theory
	provides invariance-based and derivative-free stability tools
	\cite{khalil_2002}. The present work builds on these control-theoretic foundations
	by bridging measurement-based quantum control with finite-difference Lyapunov
	methods and invariance principles, thereby aligning quantum stabilization with
	core themes of modern control theory.

	Related research has also addressed the complementary problem of quantum system
	identification and robustness under model uncertainty. Fundamental limitations
	and capabilities of identifying unknown quantum dynamics from input--output data
	have been analyzed in the context of black-box quantum systems
	\cite{burgarth_yuasa_prl_2012}. Such results highlight that, even under idealized
	conditions, system identification may only be possible up to equivalence classes
	and often requires strong structural assumptions.

	In parallel, robust stability of quantum systems subject to uncertain
	Hamiltonian perturbations has been studied within a control-theoretic framework
	\cite{petersen_ugrinovskii_james_2012}. These approaches provide valuable
	robustness guarantees but rely on explicit system models and uncertainty
	descriptions. Together, these developments underline both the relevance and the
	practical limitations of model-based robust, switching, and learning-based
	approaches, thereby motivating the pursuit of stabilization methods that operate
	without system identification and rely solely on measurement-derived information.

\noindent
	{\it Conceptual bottleneck addressed by this work.}
	Despite extensive progress in quantum control, existing stabilization methods
	share a common information-structural limitation. Lyapunov-based designs,
	stochastic feedback schemes, switching control strategies, adaptive
	identification methods, and learning-based approaches all require access to at
	least one of the following: the system generator, the quantum state, analytic
	Lyapunov derivatives, a parametric model of the dynamics, or extensive offline
	training data. These requirements are incompatible with many practical quantum
	platforms, where the generator is unknown, the state is not observable, and only
	limited measurement data are available in real time.

	The present work addresses this bottleneck by reformulating stabilization
	entirely in terms of measurement-derived finite differences of a Lyapunov
	observable. By abandoning analytic derivatives, model identification, and
	training-based policy synthesis, the proposed framework enables genuine
	model-free stabilization under severe information constraints. In particular,
	the use of sign-based finite-difference descent and adaptive gain amplification
	leads to a discrete-time analogue of LaSalle’s invariance principle that is
	applicable to information-limited quantum systems.

	To the best of our knowledge, this is the first framework that combines
	measurement-only feedback, finite-difference Lyapunov analysis, and rigorous
	invariance and ISS-type stability guarantees, without requiring system
	identification, state reconstruction, or learning-based training.

The remainder of the paper develops this framework systematically. Section~\ref{sec:Preliminaries} summarizes notation and quantum-mechanical preliminaries. Section~\ref{sec:New_Fundamental_Concepts} introduces the key structural concepts enabling model-free stabilization. Section~\ref{sec:System_Description_and_Information_Structure} formalizes the problem setting, including unknown drift, noise, and measurement constraints. Section~\ref{sec:Proposed_Framework} presents the proposed controller based on finite-difference Lyapunov descent, sign-based feedback, and adaptive gain amplification. Section~\ref{sec:Preliminary_Theoretical_Results} establishes convergence in the drift-free case, including a finite-difference LaSalle principle. Section~\ref{sec:iss} extends the analysis to unknown drift and dissipation, yielding a quantum ISS property. Section~\ref{sec:Example_Qubit_Stabilization} illustrates the method on a qubit example, and Section~\ref{sec:Conclusion} summarizes the contributions and outlines directions for future work.

\subsection{State of the Art}

The stabilization of quantum systems has traditionally been pursued through 
\emph{model-based} feedback or open-loop optimal control. In most formulations, 
the system dynamics are assumed to satisfy a known Lindblad master equation,
\begin{equation}
	\dot{\rho}(t)
	= -i[H,\rho(t)]
	+ \sum_{k}\!\left( L_k \rho(t) L_k^\dagger
	-\tfrac12 \{L_k^\dagger L_k,\rho(t)\} \right),
	\label{eq:lindblad-known}
\end{equation}
where the Hamiltonian \(H\) and dissipative operators \(\{L_k\}\) are known. 
This representation is valid under Markovian assumptions and corresponds to the 
canonical Gorini--Kossakowski--Sudarshan--Lindblad generator 
\cite{lindblad_1976,gorini_kossakowski_sudarshan_1976}. 
Most quantum-control strategies are built on this structural knowledge.

\textbf{(i) Lyapunov-based quantum control.}
Lyapunov-based stabilization--originating from early work on coherent control 
and dissipative engineering--requires explicit computation of the Lyapunov 
derivative $\dot V$ from the generator. 
A considerable body of work exploits algebraic relations among \(H\), 
$\{L_k\}$, and the Lyapunov operator to design stabilizing feedback or enforce 
convergence to decoherence-free subspaces 
\cite{ticozzi_viola_2008,ticozzi_viola_2012}. 
These approaches remain fundamentally model-dependent, as they require either 
$\dot{\rho}(t)$ or the explicit action of $\mathcal{L}$.

\textbf{(ii) Stochastic and continuous-measurement feedback.}
Continuous monitoring leads to stochastic master equations of the form
\begin{equation}
	d\rho(t)
	= \mathcal{L}(\rho(t))\,dt
	+ \mathcal{M}(\rho(t))\, dW_t,
	\label{eq:sme}
\end{equation}
introduced by Belavkin \cite{belavkin_1983} and later refined by Wiseman and 
Milburn \cite{wiseman_milburn_2010,wiseman_milburn_1993}. 
Stabilization in this setting relies on quantum filtering and the stochastic 
framework of Bouten, van Handel, and James \cite{bouten_2007}. 
These methods require full knowledge of the Lindblad generator to construct the 
filter and to design the feedback law based on the estimated state.

\textbf{(iii) Stability analysis via quantum invariance principles.}
Quantum analogues of classical invariance principles have been established for 
Markovian systems in both Schr\"odin\-ger and Heisenberg pictures 
\cite{ticozzi_viola_2012,Emzir2022,pan_amini_miao_gough_ugrinovskii_james_2014}. 
These results provide strong stability guarantees but rely on structural 
assumptions such as exact specification of $\mathcal L$, commutation relations 
(e.g.\ $[H,V]=0$), or existence of faithful invariant states. 
Such assumptions are seldom met when only limited measurement data are 
available and no model identification is feasible.

\textbf{(iv) Learning-based and adaptive Hamiltonian identification.}
Adaptive identification methods attempt to estimate unknown Hamiltonian 
parameters from measurement data \cite{zhang_sarovar_2014}. 
These schemes require a parametric model of the dynamics and informative 
measurements. 
When the number of accessible observables is small or the drift varies with 
time, such identification becomes unreliable or infeasible.

\textbf{(v) Reinforcement-learning and data-driven approaches.}
Data-driven methods, including reinforcement learning, have recently been
explored for quantum control with promising numerical results
\cite{bukov_rl_2018,niu_rl_2019}.
These approaches typically rely on extensive offline training, repeated
simulations, or implicit parametrizations of the system dynamics.
While effective in practice, they generally do not provide analytical stability
guarantees, such as Lyapunov monotonicity or invariance-based convergence
properties, and their performance may depend sensitively on the training
environment.

\bigskip

\noindent\textbf{Limitations of existing methods.}
Despite substantial progress, \emph{all existing stabilization strategies rely, 
	directly or indirectly, on knowledge of the system generator}. 
In particular, they require access to at least one of
\[
\mathcal{L},\qquad 
H,\qquad 
\{L_k\},\qquad
\dot{\rho}(t),\qquad
\nabla V(\rho),
\]
each presupposing detailed structural knowledge of the underlying dynamics. 
Classical Lyapunov methods, stochastic filtering, and quantum invariance 
principles all require evaluation of either the generator-induced derivative or 
its action on the Lyapunov function.

In realistic quantum experiments, however, only noisy measurement statistics 
are available \cite{clerk_2010}, and neither $\rho(t)$ nor $\dot{\rho}(t)$ nor 
$\mathcal{L}$ can be accessed or reconstructed reliably. 
Under these constraints, model-based stabilization techniques become 
inapplicable.

The present work takes a different approach and eliminates model dependence 
entirely. 
We rely solely on a measurement-derived Lyapunov observable and its finite 
differences, enabling stabilization without access to the generator or to the 
quantum state.

\section{Preliminaries}\label{sec:Preliminaries}

This section summarizes the notation and basic structures used throughout the paper.

\emph{Quantum states.}
A finite-dimensional quantum system with Hilbert space 
$\mathcal{H}\cong\mathbb{C}^n$ is represented by a density operator
\[
\rho \in \mathcal{D}(\mathcal{H})
:= \{X \in \mathbb{C}^{n\times n} : X \succeq 0,\; \Tr(X)=1\}.
\]
Pure states correspond to rank-one projectors $\rho=\dyad{\psi}$ for a normalized
vector $\ket{\psi}\in\mathcal{H}$.  
The state space $\mathcal{D}(\mathcal{H})$ is convex and compact, properties that
will be used to guarantee continuity and well-posedness of the model-free
feedback laws introduced later.

\emph{Quantum dynamics (unknown generator).}
The uncontrolled evolution is governed by an unknown completely positive,
trace-preserving (CPTP) generator,
\[
\dot{\rho}(t)=\mathcal{F}(\rho(t)),
\]
with no structural assumptions beyond complete positivity and trace preservation.
In particular, the controller has no knowledge of the Hamiltonian component,
dissipative operators, or whether the evolution is Markovian.

\emph{Measurement model.}
Information about the system is obtained through a fixed positive operator-valued
measure (POVM) $\{M_j\}$, where $M_j\succeq 0$ and $\sum_j M_j=I$.  
The probability of outcome $j$ at time $t$ is
\[
p_j(t)=\Tr(M_j\rho(t)).
\]
Measurement data are used only to evaluate scalar Lyapunov-like quantities
constructed from measurement statistics, without any form of state reconstruction
or model identification.

\emph{Control inputs.}
The experimenter can apply a set of Hamiltonians $\{H_k\}$ with scalar inputs
$u_k(t)$, yielding the controlled evolution
\begin{equation}
	\dot{\rho}(t)
	= \mathcal{F}(\rho(t))
	+ \sum_{k} u_k(t)\,[-iH_k,\rho(t)],
	\label{eq:controlled_dynamics}
\end{equation}
where $[-iH_k,\rho]$ denotes the unitary direction generated by $H_k$.  
The inputs $u_k(t)$ must be determined solely from measurement-derived
information; neither $\mathcal{F}$ nor $\dot{\rho}(t)$ is accessible to the
controller.

\emph{Control-theoretic interpretation.}
From a nonlinear control perspective, the problem corresponds to stabilizing an
unknown dynamical system evolving on a compact manifold.  
Only a scalar measurement-derived signal $V(t)$ is available, and the feedback
law must rely exclusively on finite-difference variations of $V(t)$.  
This setting places the framework within derivative-free Lyapunov methods and
information-limited output feedback.

The definitions introduced in this section provide the mathematical and physical
background for the conceptual framework and control problem formulated in
Sections~\ref{sec:New_Fundamental_Concepts}
and~\ref{sec:System_Description_and_Information_Structure}.

\section{New Fundamental Concepts}\label{sec:New_Fundamental_Concepts}

We introduce four structural notions that form the basis of a model-free
stabilization framework. These concepts abstract away from Hamiltonians and
Lindblad operators, focusing on what can be inferred from measurement data alone.

\begin{definition}
	A pure state $\ket{\psi}$ is \emph{model-free stabilizable} if there exists a
	feedback law
	\[
	u(t) = (u_1(t),\dots,u_m(t)),
	\]
	depending solely on accessible measurement data, such that
	\[
	\rho(t) \to \dyad{\psi}
	\qquad \text{for all }\rho(t_0)\in\mathcal{D}(\mathcal{H}),
	\]
	where $\Tr(\dyad{\psi}\rho(t))\to 1$.  
	The convergence must hold independently of the unknown generator $\mathcal{F}$.
\end{definition}

This definition formalizes the stabilization objective in the absence of any model
information.

\begin{definition}
	An observable $O(t)$, possibly updated from measurement data, is an
	\emph{adaptive Lyapunov observable} for the target $\ket{\psi}$ if
	\(
	V(t)=1-\Tr(O(t)\rho(t))\ge 0
	\)
	and $V(t)$ can be rendered decreasing under an admissible model-free
	feedback law.
\end{definition}

\begin{remark}
		The Lyapunov observable in Definition~3.2 is not assumed to be constructed from
		full state information. Instead, it is defined operationally through quantities
		directly accessible from measurement outcomes. Given a fixed POVM or a measured
		observable, the Lyapunov-like value
		\[
		V(t) = 1 - \mathrm{Tr}(O(t)\rho(t))
		\]
		is obtained from measurement statistics or classical post-processing of outcome
		frequencies, without any form of state reconstruction or model identification.

		Adaptivity of the observable $O(t)$ refers to the possibility of updating or
		selecting the Lyapunov observable using only past measurement data and knowledge
		of the target state. This may include switching among predefined observables,
		adjusting weighting coefficients, or redefining reference projectors based on
		observed descent behavior. Importantly, such updates depend solely on classical
		information derived from measurements and do not require access to $\rho(t)$,
		$\dot\rho(t)$, or the system generator.
\end{remark}

\begin{definition}
	A feedback law is a \emph{perturbation-based descent controller} if it selects
	its control direction by comparing finite-difference variations of $V(t)$
	under small probing perturbations of the control input.
\end{definition}

Such controllers require only evaluations of $V(t)$ and do not depend on
knowledge of $\mathcal{F}$ or analytic gradients.

\begin{definition}
	A feedback law is \emph{information-limited} if it depends exclusively on the
	stream of measurement outcomes obtained from a fixed POVM, possibly noisy or
	incomplete.
\end{definition}

This definition captures realistic constraints in which only a single observable
(or fixed collection of observables) is continuously accessible.


\section{System Description and Information Structure}\label{sec:System_Description_and_Information_Structure}

Let $\mathcal{H}\cong\mathbb{C}^n$ be a finite-dimensional Hilbert space, and let
$\rho(t)\in\mathcal{D}(\mathcal{H})$ denote the state at time $t$.  
The uncontrolled dynamics are governed by an \emph{unknown} CPTP generator:
\begin{equation}
	\dot{\rho}(t)=\mathcal{F}(\rho(t)),
	\label{eq:unknown-drift}
\end{equation}
with no structural assumptions beyond complete positivity and trace preservation.
Thus, for control purposes, \eqref{eq:unknown-drift} behaves as an arbitrary
nonlinear drift on the compact manifold $\mathcal{D}(\mathcal{H})$.  
The controller does not know the Hamiltonian, the dissipative terms, or whether
the evolution is Markovian.

\medskip
\noindent\textbf{Control Inputs.}
A family of Hamiltonians $\{H_k\}$ can be applied with scalar inputs $u_k(t)$,
leading to the controlled system
\begin{equation}
	\dot{\rho}(t)
	= \mathcal{F}(\rho(t))
	+ \sum_{k} u_k(t)\, G_k(\rho(t)),
	\label{eq:controlled-general}
\end{equation}
where 
\[
G_k(\rho):=[-iH_k,\rho]
\]
are known control vector fields.

	The assumption that the control Hamiltonians $\{H_k\}$ are known reflects
	standard experimental practice: while the intrinsic drift and dissipative
	dynamics contained in $\mathcal{F}$ are typically unknown or time-varying, the
	control Hamiltonians correspond to externally applied and experimentally
	calibrated control fields designed by the experimenter. Their functional form
	is therefore known by construction, even though their precise effect on the
	quantum state may be influenced by unknown drift or noise.

	Importantly, knowledge of $\{H_k\}$ does not imply access to the quantum state
	$\rho(t)$ or explicit evaluation of the vector fields $G_k(\rho(t))$. The
	controller never computes $[-iH_k,\rho(t)]$ and does not require state
	reconstruction; it relies solely on the reproducibility of the applied control
	actions and on measurement-derived evaluations of the Lyapunov observable.

This mirrors the classical structure 
\(
\dot{x}=f(x)+\sum_k u_k g_k(x),
\)
with $f$ unknown and $\{g_k\}$ known.
	In this sense, the proposed framework is model-free with respect to the system
	dynamics $\mathcal{F}$, while requiring only experimentally calibrated control
	channels, exactly as in classical nonlinear control under unknown drift.

\medskip
\noindent\textbf{Measurement Model.}
Information about the system is obtained through a fixed POVM $\{M_j\}$ with
$M_j\succeq 0$ and $\sum_j M_j=I$.  
The measurement statistics are
\[
p_j(t)=\Tr(M_j\rho(t)).
\]
The controller does \emph{not} have access to
\[
\rho(t),\qquad
\mathcal{F},\qquad
\dot{\rho}(t),\qquad
\dot{V}(t),
\]
so no observer design, model reconstruction, or derivative computation is
possible.

\medskip
\noindent\textbf{Control Objective.}
For a target pure state $\ket{\psi}$ with projector $P_\psi=\dyad{\psi}$, the
primary objective is to design a feedback law that stabilizes the system using
only measurement-derived information. In the idealized drift-free case, the
objective is asymptotic stabilization of the target state in the sense that
\begin{equation}
	\rho(t)\to P_\psi \qquad\text{as } t\to\infty.
	\label{eq:objective}
\end{equation}
In the sampled-data implementation considered in this paper, this objective is
interpreted along the sampling instants, i.e.,
\[
\rho(t_n)\to P_\psi \qquad\text{as } n\to\infty,
\]
for all initial conditions $\rho(t_0)\in\mathcal{D}(\mathcal{H})$, where
$t_n=t_0+n\tau$ denotes the sampling instants.

	In the presence of unknown Hamiltonian or dissipative contributions contained in
	$\mathcal{F}$, exact convergence may be fundamentally unattainable. In this more
	general setting, the objective is practical stabilization: the Lyapunov
	observable associated with $P_\psi$ is required to converge to, and remain
	within, a neighborhood of zero whose size depends on the magnitude of the
	unknown drift and disturbances. This notion is formalized later through a
	quantum analogue of input-to-state stability (ISS).

	Crucially, the controller is subject to severe information constraints. Neither
	the quantum state $\rho(t)$, nor the generator $\mathcal{F}$, nor analytic
	derivatives of the Lyapunov observable are available. The feedback law must be
	constructed solely from the measurement history and finite-difference variations
	of a scalar Lyapunov observable evaluated at sampled times.

\begin{problem}
	Given an unknown drift $\mathcal{F}$, known control Hamiltonians $\{H_k\}$, and
	only POVM measurement data, design an output-feedback law $u_k(t)$ based solely
	on the measurement history such that the stabilization objective
	\eqref{eq:objective} is achieved in the drift-free case, and practical
	stabilization is guaranteed in the presence of unknown drift, without
	estimating or identifying the generator $\mathcal{F}$.
\end{problem}

This formulation reflects the essential challenge: stabilizing an unknown quantum
system evolving on a nonlinear manifold using only scalar sampled-output
information, with no access to the underlying dynamics.

\section{Proposed Framework}\label{sec:Proposed_Framework}

We now introduce a model-free stabilization framework suited to the problem
formulated above. The main idea is to construct a Lyapunov-like functional from
accessible measurement data and to enforce its monotonic decrease using only
finite-difference information evaluated at discrete sampling instants.

From the available measurement scheme we extract a scalar quantity
\[
V(t)\ge 0, \qquad 
V(t)=0 \;\text{iff}\; \rho(t)=P_\psi,
\]
which serves as a Lyapunov observable.  
Typical choices include
\( V(t)=1-\Tr(P_\psi\rho(t)) \), though $V(t)$ may be adaptive or generated
directly from measurement outcomes.

	In particular, for a projective measurement
	$\{P_\psi,\, I-P_\psi\}$, the value of $V(t_n)$ corresponds operationally to the
	empirical probability of obtaining the outcome associated with the target
	projector $P_\psi$. In practice, this probability is estimated from measurement
	statistics collected over a finite sampling window preceding the sampling
	instant $t_n$, without any form of state reconstruction.

	More generally, $V(t_n)$ may be constructed from measurement statistics or
	classical post-processing of POVM outcomes, and its adaptive modification may be
	based solely on past measurement data and knowledge of the target state, for
	instance by switching among predefined observables or adjusting reference
	projectors.

Because $\dot V(t)$ is inaccessible, the controller has access only to sampled
values of the Lyapunov observable at the sampling instants
\[
t_n = t_0 + n\tau,
\]
and to finite differences computed from successive samples.

This reflects the mixed continuous--discrete information structure of the
	problem: while the quantum state evolves in continuous time according to the
	underlying dynamics, all measurements, control updates, and stability
	assessments are performed at discrete sampling instants.

To determine a descent direction, we use the finite difference
\[
\Delta V(t_n) := V(t_n)-V(t_{n-1}),
\]
which plays the role of an empirical derivative over the most recent sampling
interval.

	The sampling interval $\tau$ is a design parameter reflecting the available
	measurement rate and control bandwidth. It is chosen sufficiently large for the
	effect of a control action on $V$ to be distinguishable from measurement noise,
	yet sufficiently small to capture local descent behavior.

	Since neither the quantum state nor the generator of the dynamics is available,
	analytic evaluation of $\dot V(t)$ or $\nabla V$ is impossible in the proposed
	model-free setting. All stability guarantees are therefore formulated directly
	in terms of finite differences evaluated along the continuous-time flow over
	successive sampling intervals.

A simple model-free control law is defined in sampled-data form as
\begin{equation}
	u_k(t)
	= -\kappa_k(t_n)\,
	\operatorname{sign}\!\bigl( V(t_n)-V(t_{n-1}) \bigr),
	\qquad t\in[t_n,t_{n+1}),
	\label{eq:sign_control_law}
\end{equation}
where $\kappa_k(t_n)>0$ are adaptive gains updated at the sampling instants and
held constant between updates (zero-order hold).

	The sign-based structure ensures robustness with respect to unknown scaling of
	the system dynamics, model uncertainty, and measurement noise, as it does not
	rely on the magnitude of $\Delta V(t_n)$ but only on its sign.

The rule selects, at each sampling instant, the control direction that most
recently reduced the Lyapunov observable.

In this sense,~\eqref{eq:sign_control_law} constitutes a model-free analogue of
	classical Lyapunov descent, replacing the condition $\dot V(t)<0$ with a
	sign-consistent finite-difference criterion evaluated along the sampling
	sequence.

To ensure that the control action eventually dominates unknown drift or
measurement noise, the gains are increased whenever the observed decrease of the
Lyapunov observable over a sampling interval is insufficient. Specifically, at
each sampling instant $t_n$ the gains are updated according to
\begin{equation}
	\kappa_k(t_{n+1})
	= \kappa_k(t_n)
	+ \alpha_k\, \bigl| V(t_n)-V(t_{n-1}) \bigr|,
	\qquad \alpha_k>0,
	\label{eq:gain_update}
\end{equation}
and are held constant on the interval $[t_n,t_{n+1})$ (zero-order hold).
This mechanism parallels classical variable-gain descent and allows the
controller to ``learn'' the required actuation magnitude.

	In practice, the gains $\kappa_k(t_n)$ are initialized with small positive
	values, while the parameters $\alpha_k>0$ determine the rate of gain
	amplification, not a precise control magnitude. As a result, no prior
	knowledge of appropriate gain values is required: whenever the applied control
	is insufficient to induce finite-difference descent, the gains increase
	automatically until the effect of the control dominates unknown drift or
	disturbances.
	
	Combined with the sign-based feedback, adaptive gain amplification guarantees
	that whenever a descent direction exists at a given sampling instant, its
	effect is eventually enforced, without requiring gradient estimation or system
	identification.

Taken together, the proposed sign-based feedback and adaptive gain update
yield the following closed-loop architecture:
\begin{center}
	\begin{tabular}{c}
		Quantum system (continuous-time evolution) \\[2mm]
		$\Downarrow$ (sampled measurement at $t_n$) \\[1mm]
		Measurement record $\Rightarrow$ computation of $V(t_n)$ \\[1mm]
		$\Downarrow$ (finite-difference descent logic) \\[1mm]
		Control inputs $u_k(t)$ via \eqref{eq:sign_control_law}, \eqref{eq:gain_update} \\
		(zero-order hold on $[t_n,t_{n+1})$)
	\end{tabular}
\end{center}

This loop is:
\begin{itemize}
	\item \emph{model-free} -- no knowledge or reconstruction of $\mathcal{F}$,
	\item \emph{information-limited} -- only scalar sampled measurement data are used,
	\item \emph{derivative-free} -- descent is enforced solely from finite differences
	evaluated at the sampling instants.
\end{itemize}

	\paragraph{Practical parameter-selection and tuning guide.}
	Although the proposed framework avoids model-dependent tuning, its practical
	implementation requires selecting a small number of design parameters. The
	procedure used in the simulations, and directly applicable in experiments, is
	summarized below.
	
	\emph{Step~1 (Lyapunov observable).}
	Select a scalar observable $V(t)$ that is computable from the available
	measurement scheme and satisfies $V(t)\ge 0$ with $V(t)=0$ at the target state.
	For pure-state stabilization, a natural choice is
	$V(t)=1-\mathrm{Tr}(P_\psi\rho(t))$, estimated from measurement outcome
	statistics. No state reconstruction is required.
	
	\emph{Step~2 (Sampling interval and measurement window).}
	Choose the sampling interval $\tau$ according to the measurement rate and control
	bandwidth. In practice, $\tau$ should be large enough for control-induced changes
	in $V(t)$ over a single sampling interval to be distinguishable from measurement
	noise, yet small enough to capture local descent behavior. When $V(t)$ is
	estimated from repeated measurement shots, $\tau$ is naturally tied to the
	duration of the measurement window.
	
	\emph{Step~3 (Initialization of adaptive gains).}
	Initialize the gains $\kappa_k(t_0)$ with small positive values at the initial
	sampling instant $t_0$. The exact choice is not critical: if the applied control
	is insufficient to induce finite-difference descent between sampling instants,
	the gain amplification law~\eqref{eq:gain_update} automatically increases
	$\kappa_k(t_n)$ until a descent direction is enforced.
	
	\emph{Step~4 (Gain amplification rates).}
	Select $\alpha_k>0$ to determine the speed of gain adaptation. Larger $\alpha_k$
	lead to faster dominance over unknown drift at the cost of stronger transient
	control activity, while smaller values yield smoother but slower convergence.
	
	\emph{Step~5 (Control constraints).}
	Impose bounds $|u_k(t)| \le u_{\max}$ to reflect physical actuator limitations.
	In simulations, these bounds are explicitly enforced by saturating the control
	inputs according to $|u_k(t)| \le u_{\max}$ with $u_{\max}=2.0$, while in
	experimental implementations they are naturally imposed by hardware constraints.
	
	\emph{Step~6 (Interpretation of steady behavior).}
	In the presence of unknown drift or noise, persistent bounded oscillations of
	the sampled Lyapunov values $V(t_n)$ indicate disturbance-limited (ISS-type)
	stabilization and not improper tuning. The size of the residual neighborhood
	can be reduced by increasing the admissible bounds on the control inputs or by
	improving measurement resolution, without modifying the control structure above.


\section{Preliminary Theoretical Results}\label{sec:Preliminary_Theoretical_Results}

Before presenting the theoretical analysis, it is important to clarify the
	scope of the results in relation to existing quantum control methods. Most
	Lyapunov-based, robust, adaptive, or learning-based control strategies assume
	access to a system model, the quantum state, or analytic Lyapunov derivatives.
	In contrast, the present framework operates under strictly weaker information
	assumptions: neither the generator nor the state is known, and control decisions
	are based solely on finite differences of a measurement-derived Lyapunov
	observable. As a result, the following analysis does not aim to optimize
	performance relative to model-based baselines, but to establish stability
	guarantees that are achievable under severe information constraints.

This section establishes the fundamental stability properties of the proposed
model-free controller in the \emph{drift-free} regime. Throughout, we assume that
the intrinsic dynamics contain no unknown Hamiltonian or dissipative contribution,
so that the evolution is purely control-driven:
\[
\dot{\rho}(t)
= \sum_{k} u_k(t)\,[-iH_k,\rho(t)].
\]
Although idealized, this regime isolates the core effect of the finite-difference
descent mechanism and enables a clean convergence analysis. These results form the
basis for the ISS-type analysis in Section~\ref{sec:iss}, where unknown drift and
noise are reintroduced.

\paragraph{Continuous-time evolution vs.\ sampled information.}
The quantum state $\rho(t)$ evolves in continuous time, but the controller
receives information only at discrete sampling instants. We therefore adopt a
standard sampled-data implementation: for a fixed sampling period $\tau>0$ and an
initial sampling time $t_0$, define
\[
t_n := t_0 + n\tau,\qquad n\in\mathbb{N}.
\]
At each $t_n$ the controller evaluates the measurement-derived Lyapunov observable
$V(t_n)$ and updates the control input, which is then held constant on the
interval $[t_n,t_{n+1})$ (zero-order hold).

Let $V(t)$ denote such a measurement-derived Lyapunov observable. The only
available descent information is the sampled finite difference
\[
\Delta V(t_n) := V(t_n)-V(t_{n-1}),\qquad n\ge 1,
\]
which replaces the inaccessible derivative $\dot V(t)$.

\begin{definition}[Observable descent condition (sampling version)]
	We say that $V$ satisfies the \emph{observable descent condition} (with sampling
	period $\tau$) if there exists $N\in\mathbb{N}$ such that along the closed-loop
	sampled trajectory,
	\[
	\Delta V(t_n)<0 \quad\text{for all } n\ge N,
	\]
	except possibly when $V$ is locally constant on the sampling window
	$[t_{n-1},t_n]$ (equivalently, when $\Delta V(t_n)=0$).
\end{definition}

This replaces the classical condition $\dot V(t)<0$ with a
measurement-compatible finite-difference analogue formulated directly on the
sampling sequence $\{t_n\}$.

\begin{lemma}[Finite-difference one-step descent under double-probe (uniform level-set form)]
	\label{lemma:descent}
	Assume drift-free dynamics, i.e.\ $\mathcal{F}(\rho)=0$, so that
	\[
	\dot{\rho}(t)=\sum_{k=1}^m u_k(t)\,[-iH_k,\rho(t)].
	\]
	Let $V(\rho)$ be a continuous (possibly adaptive) Lyapunov observable and set
	$V(t):=V(\rho(t))$.
	Assume a sampled-data (zero-order hold) implementation with sampling period
	$\tau>0$ and sampling instants $t_n=t_0+n\tau$.
	
	At each sampling instant $t_n$, for each channel $k$ the controller considers two
	\emph{constant candidate inputs} of opposite sign,
	\[
	u^{+}_{k}(t)\equiv +\kappa_k(t_n),\qquad
	u^{-}_{k}(t)\equiv -\kappa_k(t_n),
	\]
	to be applied on $[t_n,t_{n+1})$, and it selects the sign that yields the smaller
	(one-step-ahead) Lyapunov value, i.e.\ it implements the double-probe rule
	\[
	u_k(t)\in\arg\min_{\sigma\in\{+,-\}}
	V\!\bigl(\Phi_{u^{\sigma}}(\tau,\rho(t_n))\bigr)
	\qquad\text{for }t\in[t_n,t_{n+1}),
	\]
	where $\Phi_{u}(\tau,\rho)$ denotes the flow at time $\tau$ under a constant input $u$.
	
	Assume moreover the following \emph{uniform one-step descendability on level sets}:
	for every $\varepsilon>0$ there exist constants $\kappa_\star(\varepsilon)>0$ and
	$\eta(\varepsilon)>0$ such that for every state $\rho$ satisfying $V(\rho)\ge\varepsilon$
	and every sampling instant $t_n$ with $\rho(t_n)=\rho$, whenever $V$ is not locally
	constant on the preceding sampling interval $[t_{n-1},t_n)$ one has
	\begin{equation}
		\min_{k\in\{1,\dots,m\}}\ \min_{\sigma\in\{+,-\}}
		V\!\bigl(\Phi_{u_k=\sigma \kappa}(\tau,\rho)\bigr)
		\ \le\ V(\rho)-\eta(\varepsilon)
		\qquad \text{for all }\kappa\ge \kappa_\star(\varepsilon).
		\label{eq:uniform_descent_levelset}
	\end{equation} 
	
	Finally, assume that the gains are updated at sampling instants by
	\[
	\kappa_k(t_{n+1})=\kappa_k(t_n)+\alpha_k\,|V(t_n)-V(t_{n-1})|,
	\qquad \alpha_k>0,
	\]
	whenever the observed decrease is insufficient (e.g.\ $V(t_n)-V(t_{n-1})\ge 0$).
	
	Then, for every $\varepsilon>0$ there exists a finite index $N_\varepsilon$ such that
	\[
	V(t_n)\le \varepsilon \qquad \text{for all } n\ge N_\varepsilon.
	\]
	In particular, $\lim_{n\to\infty}V(t_n)=0$ along the sampling instants.
\end{lemma}

\begin{proof}
	Work with a sampled-data (zero-order hold) implementation with sampling period
	$\tau>0$ and sampling instants $t_n=t_0+n\tau$. The one-step update is
	\[
	\rho(t_{n+1})=\Phi_{u(t_n)}\!\bigl(\tau,\rho(t_n)\bigr),
	\qquad V(t_n):=V(\rho(t_n)).
	\]
	
	Fix an arbitrary $\varepsilon>0$. Consider an index $n$ such that $V(t_n)\ge\varepsilon$
	and $V$ is not locally constant on $[t_{n-1},t_n)$. By the uniform level-set
	descendability assumption~\eqref{eq:uniform_descent_levelset}, there exist
	$\kappa_\star(\varepsilon)>0$ and $\eta(\varepsilon)>0$ such that for every
	$\kappa\ge \kappa_\star(\varepsilon)$ one can find a channel $k$ and a sign
	$\sigma\in\{+,-\}$ with
	\[
	V\!\bigl(\Phi_{u_k=\sigma \kappa}(\tau,\rho(t_n))\bigr)\le V(t_n)-\eta(\varepsilon).
	\]
	Since the double-probe controller selects the sign yielding the smaller one-step
	value, it achieves at least this decrease whenever the corresponding gain satisfies
	$\kappa_k(t_n)\ge \kappa_\star(\varepsilon)$. Hence, once the gains are above the
	threshold, every time the trajectory satisfies $V(t_n)\ge\varepsilon$ (and $V$ is not
	locally constant on $[t_{n-1},t_n)$) the closed loop enforces the uniform one-step decrease
	\[
	V(t_{n+1})\le V(t_n)-\eta(\varepsilon).
	\]
	
	Now use the adaptive gain update. Whenever the observed decrease is insufficient, the
	update
	\[
	\kappa_k(t_{n+1})=\kappa_k(t_n)+\alpha_k\,|V(t_n)-V(t_{n-1})|
	\]
	increases the gains. Therefore, unless the exceptional case of local constancy persists,
	each gain eventually exceeds $\kappa_\star(\varepsilon)$ after finitely many non-descent
	events.
	
	Suppose, by contradiction, that $V(t_n)\ge \varepsilon$ for infinitely many indices $n$.
	For all sufficiently large such indices the gain threshold is exceeded, so each such visit
	produces a decrease by at least $\eta(\varepsilon)$. After $r$ such visits we would obtain
	\[
	V(t_n)\le V(t_{n_0})-r\,\eta(\varepsilon),
	\]
	which is impossible for arbitrarily large $r$ because $V\ge 0$. Hence $V(t_n)\ge\varepsilon$
	can occur only finitely many times, i.e.\ there exists $N_\varepsilon$ such that
	$V(t_n)\le \varepsilon$ for all $n\ge N_\varepsilon$.
	
	Since $\varepsilon>0$ was arbitrary, it follows that $\lim_{n\to\infty}V(t_n)=0$.
\end{proof}

\begin{remark}
	This lemma highlights the central mechanism enabling model-free stabilization
	in the sampled-data setting.
	It shows that a purely measurement-driven, derivative-free update rule can
	reliably extract a descent direction from finite-difference information alone,
	provided the available control Hamiltonians generate nontrivial dynamics at the
	current sampled state.
	
	The double-probe implementation used in Section~\ref{sec:Example_Qubit_Stabilization} can be viewed as a practical realization
	of the one-step comparison in~\eqref{eq:uniform_descent_levelset}, where the sign (and, if
	desired, the channel) is selected based on finite-horizon Lyapunov evaluations under opposite
	probing actions. A pseudo-gradient variant obtained from symmetric probes leads to the same
	level-set descendability requirement and is covered by the same uniform margin hypothesis.
\end{remark}

\begin{remark}[On the role of uniform level-set descendability]
	Lemma~\ref{lemma:descent} relies on a \emph{uniform one-step descendability}
	assumption formulated on positive level sets of the Lyapunov observable.
	This assumption should be interpreted as a controllability-type requirement
	expressed in Lyapunov coordinates: away from the target set $\{V=0\}$, the
	available control directions must allow a finite-horizon decrease of $V$
	by a margin that is uniform over each level set $\{V\ge\varepsilon\}$.
	
The adaptive gain mechanism does not create descent directions; it ensures that,
whenever such directions exist, they are eventually exploited through
sufficiently large probing amplitudes. Without uniformity on level sets,
strictly positive plateau values of $V$ could not be excluded using
finite-difference information alone.

	This assumption is natural in the present information-limited setting and
	is the finite-difference analogue of the uniform descent or detectability
	conditions commonly invoked in sampled-data and input-to-state stability
	analyses. An analogous level-set uniformity hypothesis appears explicitly
	in Section~\ref{sec:iss} when establishing practical (ISS-type) stability in the presence
	of unknown drift.
\end{remark}

We now quantify how the adaptive gain mechanism prevents the system from
remaining indefinitely on any strictly positive plateau of $V$.

\begin{lemma}
	\label{lemma:gain-growth}
	Assume drift-free dynamics,
	\[
	\dot{\rho}(t)=\sum_k u_k(t)\,[-iH_k,\rho(t)].
	\]
	Let $V(t)$ be an adaptive Lyapunov observable. Suppose that the controller is
	implemented in sampled-data form with sampling period $\tau>0$, sampling
	instants $t_n:=t_0+n\tau$, and zero-order hold. The control inputs and gains are given by
	\[
	u_k(t)=u_k(t_n)
	=-\kappa_k(t_n)\,
	\operatorname{sign}\!\left(V(t_n)-V(t_{n-1})\right),
	\qquad t\in[t_n,t_{n+1}),
	\]
	\[
	\kappa_k(t_{n+1})
	=\kappa_k(t_n)+\alpha_k\,|V(t_n)-V(t_{n-1})|,
	\qquad \alpha_k>0.
	\]
	Assume further that $V$ is not eventually locally constant along the sampling
	sequence, i.e.\ $V(t_n)-V(t_{n-1})\neq 0$ for infinitely many $n$. Then:
	\begin{enumerate}
		\item If $\Delta V(t_n):=V(t_n)-V(t_{n-1})\ge 0$ occurs for infinitely many
		sampling instants, then $\kappa_k(t_n)\to\infty$.
		\item Consequently, under the drift-free dynamics and the \emph{uniform level-set}
		descent mechanism of Lemma~\ref{lemma:descent}, the sampled Lyapunov sequence
		$V(t_n)$ cannot converge to any strictly positive limit.
	\end{enumerate}
\end{lemma}

\begin{proof}
	We work with a sampled-data (zero-order hold) implementation with sampling
	period $\tau>0$. Measurements of the Lyapunov observable are available only at
	sampling instants $t_n := t_0 + n\tau$, and the finite-difference increment is
	\[
	\Delta V(t_n) := V(t_n)-V(t_{n-1}).
	\]
	
	\medskip
	\noindent
	\emph{Proof of (1).}
	If $\Delta V(t_n)\ge 0$ for infinitely many sampling instants and $V$ is not locally
	constant on the corresponding sampling windows, then $|\Delta V(t_n)|>0$ for
	infinitely many $n$. The gain recursion
	\[
	\kappa_k(t_{n+1})
	=\kappa_k(t_n)+\alpha_k|\Delta V(t_n)|,\qquad \alpha_k>0,
	\]
	implies
	\[
	\kappa_k(t_n)\ge \kappa_k(t_0)+\alpha_k\sum_{i=0}^{n-1}|\Delta V(t_i)|\to+\infty,
	\]
	because the sum contains infinitely many strictly positive terms.
	
	\medskip
	\noindent
	\emph{Proof of (2).}
	Suppose by contradiction that
	\[
	V(t_n)\to V_\star>0.
	\]
	Set $\varepsilon:=V_\star/2>0$. Then there exists $N_0$ such that
	\[
	V(t_n)\ge \varepsilon \qquad \text{for all } n\ge N_0.
	\]
	
	Since $V$ is not eventually locally constant, there are infinitely many indices
	$n\ge N_0$ for which $V$ is not locally constant on $[t_{n-1},t_n)$. Moreover, if
	$\Delta V(t_n)\ge 0$ occurred only finitely many times, then $\Delta V(t_n)<0$ for
	all sufficiently large $n$, implying that $\{V(t_n)\}$ is eventually strictly
	decreasing and hence cannot converge to a positive constant. Therefore,
	$\Delta V(t_n)\ge 0$ must occur infinitely often, and by part~(1) we obtain
	$\kappa_k(t_n)\to\infty$.
	
	Apply Lemma~\ref{lemma:descent} with the fixed level $\varepsilon>0$. It yields
	constants $\kappa_\star(\varepsilon)>0$ and $\eta(\varepsilon)>0$ such that whenever
	$V(t_n)\ge\varepsilon$ and $V$ is not locally constant on $[t_{n-1},t_n)$, any gain
	$\kappa\ge\kappa_\star(\varepsilon)$ admits a one-step decrease by at least
	$\eta(\varepsilon)$ under the double-probe selection. Since $\kappa_k(t_n)\to\infty$,
	there exists $N_1$ such that for all $n\ge N_1$ the gains exceed $\kappa_\star(\varepsilon)$.
	Hence for infinitely many indices $n\ge \max\{N_0,N_1\}$ we have
	\[
	V(t_{n+1})\le V(t_n)-\eta(\varepsilon).
	\]
	
	After $r$ such decrease events,
	\[
	V(t_n)\le V(t_{N})-r\,\eta(\varepsilon),
	\]
	which is impossible for arbitrarily large $r$ because $V(t_n)\ge 0$. This contradiction
	shows that $V_\star>0$ is impossible.
\end{proof}

\subsection{Stability Result}

We now establish asymptotic stabilization of the closed-loop system under the
assumption that the intrinsic dynamics contain \emph{no drift term}, i.e.
\[
\dot{\rho}(t)
= \sum_k u_k(t)\,[-iH_k,\rho(t)].
\]
In this setting, sufficiently large control amplitudes dominate the evolution,
which is crucial for proving strict Lyapunov descent at the sampling instants.
The analysis proceeds in four steps:
\begin{enumerate}
	\item showing that the adaptive gain mechanism prevents stagnation at any
	strictly positive Lyapunov value;
	\item proving the existence of the limit
	$\displaystyle V_\infty=\lim_{n\to\infty}V(t_n)$ along the sampling sequence;
	\item showing that no strictly positive limit is consistent with the closed-loop
	behavior;
	\item concluding that the quantum state converges to the target projector
	because $V$ is a proper Lyapunov observable.
\end{enumerate}

Lemma~\ref{lemma:gain-growth} guarantees that the Lyapunov observable cannot
remain at any strictly positive level at the sampling instants: whenever
$\Delta V(t_n)=V(t_n)-V(t_{n-1})$ fails to be negative, the gains increase,
strengthening the subsequent corrective action.
Repeated non-descent forces $\kappa_k(t_n)\to\infty$, which in the drift-free
setting ensures strict decrease of $V(t_n)$ after a finite transient.
Hence the closed-loop system cannot remain on a plateau $V_\star>0$ along the
sampling sequence.

We next show that the Lyapunov signal admits a limit along the sampling instants.

\begin{lemma}
	\label{lemma:limit-exists}
	Let $t_n=t_0+n\tau$ be the sampling instants. Assume that:
	\begin{enumerate}
		\item $0 \le V(t_n) \le V_{\max} < \infty$ for all $n$;
		\item the adaptive sign-based controller is applied under drift-free dynamics;
		\item there exists $N\in\mathbb{N}$ such that for all $n\ge N$,
		\[
		V(t_{n+1}) \le V(t_n).
		\]
	\end{enumerate}
	Then the limit
	\[
	V_\infty := \lim_{n\to\infty} V(t_n)
	\]
	exists and is finite.
\end{lemma}

\begin{proof}
	By assumption, the sequence $\{V(t_n)\}_{n\in\mathbb{N}}$ is bounded below by $0$
	and eventually nonincreasing. Hence it converges to a finite limit $V_\infty\ge 0$.
\end{proof}

\begin{remark}
	Throughout the paper, the term \emph{sign-based controller} is used to emphasize
	that only the direction of Lyapunov descent, obtained from finite-difference
	evaluations, is exploited for feedback.
	The double-probe pseudo-gradient implementation employed in the simulations
	constitutes a smooth realization of this sign-consistent descent mechanism and
	is fully consistent with the theoretical framework.
\end{remark}

We now show that the only admissible limit is zero.

\begin{theorem}
	\label{thm:V-to-zero}
	Let the controller be implemented in sampled-data form with sampling period
	$\tau>0$ and sampling instants $t_n=t_0+n\tau$. Assume that:
	\begin{enumerate}
		\item $0 \le V(t) \le V_{\max}$ for all $t\ge t_0$;
		\item the adaptive controller enforces one-step finite-difference descent at
		the sampling instants whenever $V$ is not locally constant on the preceding
		sampling window, i.e.\ whenever $V(t_n)-V(t_{n-1})\neq 0$ it eventually holds that
		\[
		V(t_{n+1})-V(t_n)<0;
		\]
		\item stagnation of the sampled sequence above any strictly positive value is
		impossible (Lem\-ma~\ref{lemma:gain-growth});
		\item the limit $V_\infty := \lim_{n\to\infty} V(t_n)$ exists
		(Lemma~\ref{lemma:limit-exists}).
	\end{enumerate}
	Then
	\[ 
	V_\infty = 0.
	\]
\end{theorem}

\begin{proof}
	Assume for contradiction that $V_\infty>0$. Then there exists $\varepsilon>0$
	and an index $N$ such that $V(t_n)\ge V_\infty-\varepsilon>0$ for all $n\ge N$.
	By Lemma~\ref{lemma:gain-growth}, the sampled Lyapunov sequence cannot converge
	to a strictly positive plateau under the adaptive mechanism, i.e.\ stagnation
	above any positive level is impossible. This contradicts $V(t_n)\to V_\infty>0$.
	Hence $V_\infty=0$.
\end{proof}

The previous results establish that, under drift--free dynamics and the adaptive
sign-based controller, the Lyapunov observable converges to zero along the
sampling instants,
\[
\lim_{n\to\infty} V(t_n)=0, \qquad t_n =t_0+ n\tau .
\]
Since the closed-loop evolution is continuous on each sampling interval
$[t_n,t_{n+1})$ under zero-order hold, this guarantees asymptotic stabilization
of the quantum system at the sampling times. In particular, the state satisfies
\(
\rho(t_n)\to \dyad{\psi}
\)
as $n\to\infty$ whenever the Lyapunov observable is proper.

\begin{remark}
The requirement that the adaptive gains become ``sufficiently large'' should be interpreted in a control-theoretic sense. It does not imply that arbitrarily large or physically unrealistic control amplitudes are available. Instead, it asserts the existence of a gain threshold above which the control-induced variation of the Lyapunov observable dominates the unknown drift or disturbance whenever such domination is physically feasible.
	
In realistic quantum systems, control amplitudes are always bounded by hardware constraints. The present analysis is compatible with such bounds: if the admissible control amplitudes exceed the threshold required for Lyapunov descent, asymptotic stabilization is achieved in the drift-free case; otherwise, the closed-loop behavior naturally transitions to practical (ISS-type) stabilization, as analyzed in Section~\ref{sec:iss}.

\end{remark}

\begin{theorem}[Asymptotic Stabilization in the Drift--Free Case]
	\label{thm:stabilization-driftfree}
	Assume drift--free dynamics
	\[
	\dot{\rho}(t)=\sum_k u_k(t)\,[-iH_k,\rho(t)],
	\]
	and suppose that:
	\begin{enumerate}
		\item the available Hamiltonians $\{H_k\}$ generate nontrivial control
		directions at every $\rho\neq P_\psi$;
		\item $V(t)$ is a proper Lyapunov observable, i.e.,
		\[
		V(t)=0 \quad \text{if and only if} \quad \rho(t)=P_\psi;
		\]
		\item the adaptive sign-based controller enforces finite-difference descent
		whenever $V$ is not locally constant on a sampling interval.
	\end{enumerate}
	Then, along the sampling instants $t_n=t_0+n\tau$,
	\[
	\lim_{n\to\infty} V(t_n)=0,
	\]
	and consequently
	\[
	\rho(t_n)\to P_\psi \qquad \text{as } n\to\infty.
	\]
\end{theorem}

\begin{proof}
	By Lemma~\ref{lemma:limit-exists}, the limit
	$V_\infty=\lim_{n\to\infty}V(t_n)$ exists and is finite.
	By Theorem~\ref{thm:V-to-zero}, one has $V_\infty=0$, i.e., $V(t_n)\to 0$.
	Since $V(\rho)$ is continuous on the compact set $\mathcal D(\mathcal H)$ and
	$V(\rho)=0$ iff $\rho=P_\psi$, it follows that $V(t_n)\to0$ implies
	$\rho(t_n)\to P_\psi$.
\end{proof}

\begin{corollary}[Stabilization under Projective Measurement]
	Let the measurement be the projective POVM
	$\{P_\psi,\, I-P_\psi\}$ with $P_\psi=\dyad{\psi}$, and define
	$V(t)=1-\Tr(P_\psi\rho(t))$.
	Under drift--free dynamics and the adaptive sampled-data controller based on
	one-step comparison of the two candidate inputs $\pm \kappa_k(t_n)$,
	$\rho(t_n)\to P_\psi$ as $n\to\infty$.
\end{corollary}

\begin{proof}
	The observable $V(\rho)=1-\Tr(P_\psi\rho)$ is continuous, bounded, and proper.
	Moreover, the candidate-comparison rule satisfies the one-step descent property
	of Lemma~\ref{lemma:descent} whenever $V$ is not locally constant on a sampling
	window. Hence the assumptions of
	Theorem~\ref{thm:stabilization-driftfree} hold and the claim follows.
\end{proof}

We now show that the stabilization mechanism is robust to bounded measurement
errors. While persistent measurement corruption prevents guaranteeing exact
asymptotic convergence, the adaptive sign-based controller still ensures
practical stabilization to a noise-dependent neighborhood of the target.

\begin{proposition}[Robustness to Bounded Measurement Errors]
	\label{prop:robustness-bounded}
	Assume drift-free dynamics and a sampled-data implementation with sampling
	period $\tau>0$ and sampling instants $t_n=t_0+n\tau$ (zero-order hold on
	$[t_n,t_{n+1})$).
	Let the controller use the perturbed measurement
	\[
	\widetilde V(t_n)=V(t_n)+\eta(t_n),
	\qquad |\eta(t_n)|\le \eta_{\max},
	\]
	and define the noisy finite difference
	\[
	\Delta\widetilde V(t_n)
	=\widetilde V(t_n)-\widetilde V(t_{n-1}).
	\]
	Under the noisy sign-based controller
	\[
	u_k(t_n)
	= -\kappa_k(t_n)\operatorname{sign}\!\bigl(\Delta\widetilde V(t_n)\bigr),
	\qquad
	\kappa_k(t_{n+1})
	= \kappa_k(t_n)+\alpha_k|\Delta\widetilde V(t_n)|,
	\]
	assume moreover that the drift-free one-step descent mechanism of
	Lemma~\ref{lemma:descent} holds in the following local form:
	there exist constants $\varepsilon>0$ and a gain threshold $\kappa^\star>0$
	such that whenever $\kappa_k(t_n)\ge\kappa^\star$ and $V$ is not locally constant
	on $[t_{n-1},t_n)$, at least one of the two opposite constant inputs
	$\pm \kappa_k(t_n)$ yields a one-step decrease of magnitude at least $\varepsilon$
	in the noiseless Lyapunov value over $[t_n,t_{n+1})$.
	
	Then there exists a constant $C>0$ (depending on $\tau$ and the controller
	parameters, and on the local descent margin $\varepsilon$) such that the sampled
	Lyapunov values satisfy the practical bound
	\[
	\limsup_{n\to\infty} V(t_n)\;\le\; C\,\eta_{\max}.
	\]
	In particular, for sufficiently accurate measurements ($\eta_{\max}$ small),
	the residual neighborhood can be made arbitrarily small; and it can also be
	reduced by enlarging the admissible bounds on the control inputs (when such
	bounds allow a larger effective one-step descent margin).
\end{proposition}

\begin{proof}
	First note that the measurement error perturbs the finite difference by at most
	\[
	|\Delta\widetilde V(t_n)-\Delta V(t_n)|
	\le |\eta(t_n)|+|\eta(t_{n-1})|
	\le 2\eta_{\max}.
	\]
	Hence the sign of $\Delta\widetilde V(t_n)$ can differ from the sign of
	$\Delta V(t_n)$ only when the true finite-difference variation is small, i.e.,
	when $|\Delta V(t_n)|\le 2\eta_{\max}$.
	
	Consider a sampling instant $t_n$ for which the gain is already above the
	threshold, $\kappa_k(t_n)\ge\kappa^\star$, and $V$ is not locally constant on
	$[t_{n-1},t_n)$. By the assumed one-step descent property (Lemma~\ref{lemma:descent}
	in local margin form), among the two opposite candidate inputs $\pm \kappa_k(t_n)$
	there exists a choice that would decrease the noiseless Lyapunov value by at
	least $\varepsilon$ over one sampling interval.
	
	Due to the bounded measurement corruption, an incorrect sign selection can occur
	only when the observable change is masked by noise; in particular, once the
	control-induced one-step descent margin dominates the worst-case perturbation,
	i.e.\ once $\varepsilon>2\eta_{\max}$, the noisy sign rule selects a descent
	direction consistently and enforces a strict decrease of the noiseless values
	$V(t_n)$.
	
	When $\varepsilon$ is comparable to $2\eta_{\max}$, sign errors may still occur,
	but the adaptive gain update increases $\kappa_k(t_n)$ whenever the observed
	decrease is insufficient. Therefore the closed-loop trajectory cannot remain
	indefinitely in a region where $V(t_n)$ is large while the control-induced
	variation stays below the noise level. As a consequence, there exists a constant
	$C>0$ such that whenever $V(t_n)>C\eta_{\max}$, the controller enforces a net
	one-step decrease that drives $V(t_n)$ back toward the band $V\le C\eta_{\max}$.
	This yields the practical ultimate bound
	$\limsup_{n\to\infty}V(t_n)\le C\eta_{\max}$.
\end{proof}

\subsection{Finite-Difference LaSalle Principle for Model-Free Quantum Systems}

In classical nonlinear control, LaSalle’s invariance principle is a fundamental
tool for establishing asymptotic stability by analyzing the derivative of a
Lyapunov function along system trajectories.  
In the model-free quantum setting considered here, this approach is no longer
available: the generator $\mathcal{F}$ is unknown, the derivative $\dot V(t)$
cannot be evaluated, measurement data are available only at discrete sampling
times, and noise precludes reliable derivative estimation.  
Consequently, the classical LaSalle framework cannot be applied directly.

This subsection develops a model-free analogue, formulated entirely in terms of
finite differences evaluated at sampling instants.  
The central idea is the following: if a measurement-derived Lyapunov observable
exhibits strict finite-difference descent whenever the system lies outside a
designated invariant set, then the closed-loop state must converge to that set,
even in the absence of model knowledge, state reconstruction, or derivative
information.

For a deterministic system $\dot x=f(x)$ with Lyapunov function $V(x)$, the
classical LaSalle invariance principle asserts that if $\dot V(x)\le 0$ for all
$x$, then every trajectory approaches the largest invariant set contained in
$\{x:\dot V(x)=0\}$.  
Here, no such derivative-based characterization is available. Instead, the
analysis must rely on observable variations across successive sampling
intervals. In the present model-free quantum setting:
\begin{enumerate}
	\item the generator $\mathcal{F}$ of $\dot\rho=\mathcal{F}(\rho)$ is unknown;
	\item the quantum state $\rho(t)$ is unobserved and never reconstructed;
	\item measurements provide only a scalar observable $V(t)$ at discrete times
	$t_n=t_0+n\tau$;
	\item noise and sampling preclude reliable estimation of $\dot V(t)$;
	\item only the finite difference
	\[
	\Delta V(t_n)=V(t_n)-V(t_{n-1})
	\]
	is available, and the control law uses only its sign.
\end{enumerate}

These constraints motivate a LaSalle-type convergence result formulated
entirely in terms of finite-difference information. Instead of identifying
invariant sets through vanishing derivatives, the proposed principle
characterizes convergence by ruling out the persistence of strictly positive
Lyapunov plateaus under adaptive, measurement-driven descent.

\begin{theorem}[Finite-Difference LaSalle Principle (Sampling Version)]
	\label{thm:fd-LaSalle}
	Let $t_n:=t_0+n\tau$ be the sampling instants, and let
	\[
	V_n := V(t_n)
	\]
	denote the measured Lyapunov observable evaluated at sampling times.
	Assume that:
	\begin{enumerate}
		\item $0\le V_n\le V_{\max}$ for all $n$;
		\item the adaptive sign-based controller is implemented in sampled-data form
		with zero-order hold on each interval $[t_n,t_{n+1})$;
		\item for every $\rho\notin\mathcal{I}$, where
		\[
		\mathcal{I}:=\{\rho:\;V(\rho)=0\},
		\]
		the closed-loop law enforces finite-difference descent in the sense that,
		whenever $V$ is not locally constant on the sampling window
		$[t_n,t_{n+1})$, one has
		\[
		V_{n+1}-V_n<0;
		\]
		\item stagnation above any strictly positive value is impossible
		(Lemma~\ref{lemma:gain-growth}).
	\end{enumerate}
	Then the limit
	\[
	V_\infty:=\lim_{n\to\infty}V_n
	\]
	exists and satisfies $V_\infty=0$. In particular,
	\[
	\lim_{n\to\infty}\mathrm{dist}\bigl(\rho(t_n),\mathcal{I}\bigr)=0.
	\]
	If $V$ is proper (i.e.\ $V(\rho)=0$ iff $\rho\in\mathcal{I}$), then
	\[
	\lim_{n\to\infty}\rho(t_n)\in\mathcal{I},
	\qquad\text{equivalently,}\qquad
	\mathrm{dist}\bigl(\rho(t_n),\mathcal{I}\bigr)\to 0.
	\]
\end{theorem}

\begin{proof}
	By Lemma~\ref{lemma:limit-exists}, the bounded sequence $\{V(t_n)\}$ admits a
	finite limit $V_\infty\ge 0$.
	
	Suppose, by contradiction, that $V_\infty>0$. Then there exists $\delta>0$ and
	$N$ such that
	\[
	V(t_n)\ge \delta>0
	\qquad\text{for all } n\ge N,
	\]
	so $\rho(t_n)\notin \mathcal{I}$ for all $n\ge N$.
	
	Moreover, since stagnation above any strictly positive value is impossible
	(Lemma~\ref{lemma:gain-growth}), the controller cannot remain indefinitely on a
	positive plateau. In particular, for infinitely many indices $n\ge N$, the
	closed-loop evolution must produce a strict descent event over one sampling
	interval. By the descent property (assumption~3) together with the gain-growth
	mechanism, there exist $\varepsilon>0$ and infinitely many indices $n\ge N$ such
	that
	\[
	V(t_{n+1}) \le V(t_n)-\varepsilon.
	\]
	Such uniform decreases cannot occur infinitely often for a nonnegative bounded
	sequence converging to a strictly positive limit. This contradiction implies
	$V_\infty=0$.
	
	Finally, since $\mathcal{I}=\{\rho:V(\rho)=0\}$ is closed and $V(t_n)\to 0$,
	we obtain $\mathrm{dist}(\rho(t_n),\mathcal{I})\to 0$. If $V$ is proper, this
	means $\rho(t_n)\to\mathcal{I}$.
\end{proof}

This result provides a discrete-time, measurement-driven analogue of LaSalle’s
invariance principle.  
Whenever the closed-loop system remains outside the zero set of $V$ at the
sampling instants $t_n$, adaptive gain amplification guarantees that the finite
difference
\[
\Delta V(t_n)=V(t_n)-V(t_{n-1})
\]
becomes strictly negative after a finite transient.  
As a consequence, the sampled sequence $\{V(t_n)\}$ converges to zero, and the
quantum state approaches the desired invariant set at the sampling times, despite
the complete absence of model knowledge, state access, or derivative information.

\subsection{Discussion}

The stability results derived above apply to the \emph{drift-free} setting, in
which the intrinsic evolution vanishes and the state evolves solely through the
applied control Hamiltonians.  
Under this assumption, the closed-loop dynamics contain no unknown autonomous
term, and the controller interacts with the system only through the known
directions generated by $\{H_k\}$ and through sampled measurement data.

Within this framework, the analysis establishes a fully self-contained
model-free stabilization mechanism for finite-dimensional quantum systems at
the sampling instants.  
Its essential ingredients are:
\begin{itemize}
	\item[] sign-based output feedback, which enforces empirical Lyapunov descent
	using only finite-difference information;
	\item[] adaptive gain amplification, which guarantees that stagnation at any
	strictly positive Lyapunov level cannot persist;
	\item[] convergence of the sampled Lyapunov sequence $\{V(t_n)\}$, established
	without access to derivatives or any part of the generator;
	\item[] uniqueness of the limiting value $V_\infty$, which must equal zero by
	the finite-difference LaSalle principle;
	\item[] asymptotic convergence of the quantum state to the target pure state
	at the sampling times whenever the Lyapunov observable is proper.
\end{itemize}

These conclusions require no knowledge of the dynamical generator and rely only
on sampled observable data.  
The resulting stabilization law can be interpreted as a discrete-time analogue
of LaSalle’s invariance principle, augmented with adaptive feedback to guarantee
strict finite-difference descent when necessary.  
The robustness result further shows that the mechanism preserves its qualitative
behavior under bounded measurement noise, indicating compatibility with
realistic experimental uncertainty and suggesting natural connections to
stochastic and practical Lyapunov stability concepts.

\section{ISS-Type Stability Under Unknown Drift and Dissipation}
\label{sec:iss}

The stability results of the previous section were derived under the
\emph{drift-free} assumption, where the state evolves solely under the applied
control Hamiltonians.  
In this regime, the adaptive sign-based controller enforces strict
finite-difference descent of the Lyapunov observable and guarantees asymptotic
convergence to the target state.

We now consider the physically realistic case in which the dynamics include an
\emph{unknown, persistent drift} and possibly irreversible dissipative noise.
For notational clarity--and only to ensure physical consistency--we represent the
unknown generator in the Lindblad form, without assuming Markovianity or any
specific structure accessible to the controller,
\[
\mathcal{F}_{\mathrm{drift}}+\mathcal{F}_{\mathrm{noise}} \neq 0,
\]
where
\[
\mathcal{F}_{\mathrm{drift}}(\rho)=-i[H_{\mathrm{drift}},\rho],\qquad
\mathcal{F}_{\mathrm{noise}}(\rho)
=\sum_{k}\!\left(
L_k \rho L_k^\dagger
-\tfrac12\{L_k^\dagger L_k,\rho\}
\right),
\]
with $H_{\mathrm{drift}}$ and the operators $\{L_k\}$ completely unknown.  
This representation is purely formal: the model-free controller does not use or
identify any part of the generator.  
Its sole purpose is to guarantee that the unobserved dynamics correspond to a
valid CPTP evolution.

The controlled dynamics therefore satisfy
\[
\dot{\rho}(t)
=\mathcal{F}_{\mathrm{drift}}(\rho(t))
+\mathcal{F}_{\mathrm{noise}}(\rho(t))
+\sum_k u_k(t)\,[-iH_k,\rho(t)].
\]
Unknown drift and dissipation act as persistent disturbances that no
model-free controller can exactly cancel.  
Consequently, exact asymptotic stabilization is generically impossible; the
appropriate performance benchmark is \emph{input-to-state stability} (ISS), in
which the deviation from the target is bounded by a function of the disturbance
magnitude.

This mirrors the classical nonlinear setting, where persistent disturbances
prevent asymptotic regulation and the best achievable guarantee is ISS or one of
its practical variants.  
Here the drift and noise terms play the role of unknown exogenous disturbances
entering through an uncontrollable channel.  
Because the controller does not know (and cannot identify) these terms, one can seek only an ISS-type estimate of the form:
\[
V(t)\le \beta(V(0),t)+\gamma(D),
\]
where $D$ bounds the disturbance strength, $\beta\in\mathcal{KL}$, and
$\gamma\in\mathcal{K}$.  
This establishes \emph{practical model-free stabilization}: the state converges
to a neighborhood whose radius depends continuously on the disturbance level and
shrinks to zero as the disturbance vanishes.

Finally, we show that this ISS limitation is fundamental: when unknown
Hamiltonian drift is present, the target pure state is generically not an
equilibrium of the closed-loop dynamics and cannot be globally stabilized by
any model-free feedback law based solely on measurement-derived information.

\begin{lemma}[ISS-type practical stability under unknown drift]
	\label{lemma:quantum-iss}
	Consider the sampled-data closed-loop dynamics
	\[
	\dot{\rho}(t)
	=\mathcal{F}_{\mathrm{drift}}(\rho(t))
	+\mathcal{F}_{\mathrm{noise}}(\rho(t))
	+\sum_{k} u_k(t)\,[-iH_k,\rho(t)],
	\]
	with sampling period $\tau>0$ and sampling instants $t_n=t_0+n\tau$.
	Assume a zero-order hold implementation on each interval $[t_n,t_{n+1})$.
	
	Let $V(\rho)$ be a measurement-derived Lyapunov observable and define the
	one-step sampled increment $\Delta V(t_n):=V(t_n)-V(t_{n-1})$.
	Assume the control channels admit two opposite constant candidate inputs
	$\pm\kappa_k(t_n)$ (applied over $[t_n,t_{n+1})$) with adaptive gains
	\[
	\kappa_k(t_{n+1})
	=\kappa_k(t_n)+\alpha_k\,|\Delta V(t_n)|,\qquad \alpha_k>0.
	\]
	
	Assume:
	\begin{enumerate}
		\item the target is reachable under the available Hamiltonians $\{H_k\}$;
		\item $V(\rho)$ is proper and bounded on $\mathcal{D}(\mathcal{H})$,
		$0\le V(\rho)\le V_{\max}$;
		\item the disturbance satisfies the uniform bound (e.g.\ in trace norm)
		\[
		\bigl\|\mathcal{F}_{\mathrm{drift}}(\rho)+\mathcal{F}_{\mathrm{noise}}(\rho)\bigr\|
		\le D
		\quad\text{for all }\rho\in\mathcal{D}(\mathcal{H});
		\]
		\item the disturbance-induced contribution to the Lyapunov increment satisfies
		\[
		|\Delta V_{\mathrm{dist}}(t_n)| \le c\,D\,\tau,
		\]
		for some constant $c>0$ (a Lipschitz constant of $V$ on the compact state space);
		\item (\emph{Uniform one-step controllable descent on level sets})
		for every $\varepsilon>0$ there exist constants
		$\kappa^\star(\varepsilon)>0$ and $\eta(\varepsilon)>0$ such that for any
		$\rho$ with $V(\rho)\ge \varepsilon$ and any sampling instant $t_n$ with
		$\rho(t_n)=\rho$, whenever $\kappa_k(t_n)\ge \kappa^\star(\varepsilon)$,
		\begin{equation}
			\min_{\sigma\in\{+,-\}}
			\Bigl(
			V\bigl(\Phi_{u_k=\sigma\kappa_k(t_n)}(\tau,\rho(t_n))\bigr) - V(\rho(t_n))
			\Bigr)
			\le -\eta(\varepsilon).
			\label{eq:iss-min-descent}
		\end{equation}
	\end{enumerate}
	
	Then there exists a constant $C>0$ such that
	\[
	\limsup_{n\to\infty} V(t_n)\le C\,D\,\tau.
	\]
\end{lemma}

\begin{proof}
	We analyze the sampled evolution along $t_n=t_0+n\tau$.
	
	\emph{Step 1: Decomposition of the sampled increment.}
	For each sampling step,
	\[
	\Delta V(t_n)=V(t_n)-V(t_{n-1})
	=\Delta V_{\mathrm{ctrl}}(t_n)+\Delta V_{\mathrm{dist}}(t_n),
	\]
	where $\Delta V_{\mathrm{dist}}(t_n)$ collects the net effect of
	$\mathcal{F}_{\mathrm{drift}}+\mathcal{F}_{\mathrm{noise}}$ over $[t_{n-1},t_n)$,
	and $\Delta V_{\mathrm{ctrl}}(t_n)$ is the net contribution attributable to the
	(control-driven) part of the evolution under the implemented input on that
	interval.
	
	\emph{Step 2: Uniform \emph{existence} of control-induced descent on level sets.}
	Fix $\varepsilon>0$ and consider the compact level set
	\[
	\Omega_\varepsilon:=\{\rho\in\mathcal{D}(\mathcal{H}) : V(\rho)\ge \varepsilon\}.
	\]
	By Assumption~5, there exist $\kappa^\star(\varepsilon)>0$ and $\eta(\varepsilon)>0$
	such that whenever $\rho(t_n)\in\Omega_\varepsilon$ and $\kappa_k(t_n)\ge \kappa^\star(\varepsilon)$,
	\eqref{eq:iss-min-descent} holds: among the two opposite constant inputs
	$\pm \kappa_k(t_n)$ applied over $[t_n,t_{n+1})$, at least one yields a one-step
	decrease in the noiseless Lyapunov value by at least $\eta(\varepsilon)$.
	
	\emph{Step 3: Competition between the favorable control action and the disturbance.}
	By Assumption~4,
	\[
	|\Delta V_{\mathrm{dist}}(t_n)|\le c\,D\,\tau.
	\]
	Therefore, whenever $\rho(t_n)\in\Omega_\varepsilon$ and the gain is above the
	threshold, if the favorable sign in \eqref{eq:iss-min-descent} is applied on
	$[t_n,t_{n+1})$, then the total one-step change satisfies
	\[
	\Delta V(t_{n+1})
	\le -\eta(\varepsilon)+c\,D\,\tau.
	\]
	In particular, if $\eta(\varepsilon)>c\,D\,\tau$, then a strict net decrease is
	available outside $\{V<\varepsilon\}$.
	
	\emph{Step 4: Adaptive gain amplification and recurrence of descent steps.}
	If $\Delta V(t_n)\ge 0$ at some sampling instants, the adaptive update
	\[
	\kappa_k(t_{n+1})=\kappa_k(t_n)+\alpha_k|\Delta V(t_n)|
	\]
	increases the gain. Hence, for any fixed $\varepsilon>0$, the gain eventually
	exceeds $\kappa^\star(\varepsilon)$ unless the trajectory enters $\{V<\varepsilon\}$
	and remains there.
	Once $\kappa_k(t_n)\ge\kappa^\star(\varepsilon)$ holds, the existence of a strict
	net decrease outside $\{V<\varepsilon\}$ from Step~3 implies that the closed loop
	cannot indefinitely persist in $\Omega_\varepsilon$ while maintaining
	$\Delta V(t_n)\ge 0$ frequently: either it enters $\{V<\varepsilon\}$, or it
	experiences descent steps that drive it toward this set.
	
	\emph{Step 5: Ultimate boundedness.}
	Choose $\varepsilon:=C\,D\,\tau$ with $C>0$ large enough such that
	\[
	\eta(\varepsilon)>c\,D\,\tau.
	\]
	(Existence of such $C$ follows from Assumption~5, which provides a positive
	descent margin on every strictly positive level set.)
	Then, whenever $V(t_n)>\varepsilon$ and the gain is above the corresponding
	threshold, there exists a control polarity that yields a net decrease in $V$ over
	one sampling step. By Step~4 and the gain adaptation mechanism, the closed-loop
	trajectory cannot remain above $\varepsilon$ indefinitely. Therefore,
	\[
	\limsup_{n\to\infty} V(t_n)\le \varepsilon = C\,D\,\tau,
	\]
	which establishes ISS-type practical stabilization at the sampling instants.
\end{proof}

\begin{remark}[Applicability to Double-Probe Gradient Estimation]
	The same ISS-type bound extends directly to the double-probe pseudo-gradient
	controller.
	In this case, the descent direction is estimated from symmetric finite-difference
	evaluations of the Lyapunov observable under opposite constant probing actions
	applied over one sampling interval.
	Whenever the disturbance-induced variation over a sampling interval remains
	sufficiently small relative to the probing amplitude, the resulting descent
	direction is selected consistently.
	Consequently, the closed-loop evolution satisfies
	\[
	\limsup_{n\to\infty} V(t_n)\le \gamma(D),
	\]
	in agreement with the behavior observed in numerical simulations.
\end{remark}

\paragraph{ISS Interpretation and the Case $V_\infty>0$.}

The ISS analysis shows that, in the presence of unknown Hamiltonian drift or
irreversible Lindblad noise, the model-free controller cannot perfectly cancel
the disturbance.
As a result, the closed-loop Lyapunov observable does not converge to zero but
approaches a disturbance-limited neighborhood of the origin.
Lemma~\ref{lemma:quantum-iss} yields
\[
\limsup_{n\to\infty} V(t_n)\le \gamma(D),
\]
where $D$ bounds the magnitude of
$\mathcal{F}_{\mathrm{drift}}+\mathcal{F}_{\mathrm{noise}}$.
Thus a strictly positive limiting value is not a numerical artifact but a direct
consequence of ISS-type behavior:
\[
V_\infty>0
\quad\text{for generic unknown drift or noise terms.}
\]
Exact asymptotic stabilization is therefore achievable only when the disturbance
vanishes or can be compensated through additional model-based control.

\subsection{Fundamental Limitation: No Asymptotic Stabilization Without Drift Cancellation}

Unknown Hamiltonian drift acts as a persistent disturbance that induces a
continuous unitary rotation of the state.  
Since a model-free controller observes only past values of a scalar Lyapunov
observable, it cannot estimate or cancel this rotation.  
This leads to a fundamental obstruction to asymptotic stabilization.

\begin{theorem}
	\label{thm:no-asymptotic}
	Consider the closed-loop evolution
	\[
	\dot{\rho}(t)
	= -i[H_{\mathrm{drift}}+H_u(t),\,\rho(t)],
	\]
	where $H_{\mathrm{drift}}\neq 0$ is unknown and $H_u(t)$ is generated solely
	from past values of a Lyapunov observable $V(t)$. Then:
	\begin{enumerate}
		\item The drift-induced flow
		\[
		\rho(t)=e^{-iH_{\mathrm{drift}}(t-t_0)}\,\rho(t_0)\,e^{iH_{\mathrm{drift}}(t-t_0)}
		\]
		admits fixed points only for states satisfying
		\(
		[\rho,H_{\mathrm{drift}}]=0.
		\)
		\item Because only $V(t)$ is observed, the controller cannot in general
		reconstruct $H_{\mathrm{drift}}$ or synthesize a cancelling control
		$H_u(t)\approx -H_{\mathrm{drift}}$.
		\item If the target $\rho^\star$ does not commute with $H_{\mathrm{drift}}$,
		it is not an equilibrium of the closed-loop system.
	\end{enumerate}
	Consequently, no model-free controller based solely on output measurements can
	in general guarantee
	\(
	\rho(t)\to \rho^\star
	\)
	in the presence of an unknown drift.
	The strongest achievable performance is ISS-type practical stabilization, as
	formalized in Lemma~\ref{lemma:quantum-iss}.
\end{theorem}

\begin{proof}
	If $H_{\mathrm{drift}}\neq 0$, the free unitary trajectory is nontrivial unless
	the target commutes with $H_{\mathrm{drift}}$.  
	The controller receives only past scalar values of $V(t)$ and no information
	about the generator or $\dot{\rho}(t)$; hence reconstruction of
	$H_{\mathrm{drift}}$ and synthesis of a cancelling control are impossible.
	When $\rho^\star$ is not an equilibrium of the closed-loop dynamics, standard
	invariance arguments rule out asymptotic convergence.  
	ISS-type bounds therefore constitute the maximal achievable guarantee.
\end{proof}

\section{Representative Example: Qubit Stabilization}\label{sec:Example_Qubit_Stabilization}

We illustrate the model-free stabilization framework on the simplest
nontrivial system: a single qubit.  
This example demonstrates how finite-difference feedback stabilizes the target
state using only measurement data and without any knowledge of the drift
Hamiltonian. The essential closed-loop phenomena--adaptive gain amplification, oscillatory
finite-difference behavior induced by the unknown drift, and the resulting
ISS-type convergence to a disturbance-dependent neighborhood--are already fully
visible in this two-dimensional case.
 
The simulation parameters used in this section were selected following the
	practical parameter-selection and tuning procedure described in
	Section~\ref{sec:Proposed_Framework}. In particular, the sampling interval,
	initial gains, gain-adaptation rates, and control bounds were chosen according
	to measurement resolution and admissible control amplitudes, without any
	model-dependent optimization.

The target state is
\[
\ket{\psi}=\ket{0}, 
\qquad 
P_0=\dyad{0}
= 
\begin{pmatrix}
	1 & 0\\
	0 & 0
\end{pmatrix},
\]
and the complementary projector is
\[
P_1=\dyad{1}
=
\begin{pmatrix}
	0 & 0\\
	0 & 1
\end{pmatrix}.
\]
The system is measured by the projective POVM $\{P_0,P_1\}$.  
This yields the proper Lyapunov observable
\[
V(t)=1-\Tr(P_0\rho(t))=\Tr(P_1\rho(t)),
\]
which equals the population of the excited state $\ket{1}$ and therefore
quantifies the portion of the state outside the target.

\paragraph{Measurement-based evaluation of the Lyapunov observable and time scaling.}
Although the quantum state $\rho(t)$ is not accessible to the controller, the Lyapunov observable
	\(
	V(t) = 1 - \mathrm{Tr}(P_{0}\rho(t))
	\)
is directly obtainable from measurement outcomes. For the projective measurement $\{P_{0},\, P_{1}\}$, the quantity $\mathrm{Tr}(P_{0}\rho(t))$ corresponds to the probability of observing the outcome associated with the target state $|0\rangle$. In an experimental implementation, this probability is estimated from measurement statistics (e.g.\ relative frequencies collected over a finite sampling window), yielding an empirical estimate of $V(t)$ without any form of state reconstruction. In the numerical simulations, $\mathrm{Tr}(P_{0}\rho(t))$ is computed directly from $\rho(t)$ for simplicity and to avoid additional sampling noise, while preserving the same information structure available to the controller.
	
The time axis in the simulations is expressed in normalized (dimensionless) units determined by the chosen scaling of the Hamiltonians and control amplitudes. Mapping these units to physical time depends on the specific experimental platform and calibration parameters, such as qubit frequency scales and maximum achievable control strengths.

The available controls are the Pauli rotations
\[
H_x=\tfrac12\sigma_x,\qquad 
H_y=\tfrac12\sigma_y,
\]
which generate $\mathfrak{su}(2)$ and render the target reachable
(see, e.g.,~\cite{altafini_ticozzi_2012,d_alessandro_2007}).

To validate the theoretical ISS predictions, we include an \emph{unknown}
Hamiltonian drift
\[
\mathcal{F}_{\mathrm{drift}}(\rho)=-i[H_{\mathrm{drift}},\rho],
\qquad
H_{\mathrm{drift}}=0.35\,\sigma_x + 0.20\,\sigma_y + 0.45\,\sigma_z,
\]
with no dissipative noise.  
Thus the true (continuous-time) dynamics are
\[
\dot{\rho}(t)
= -i\bigl[H_{\mathrm{drift}} + u_x(t)H_x + u_y(t)H_y,\rho(t)\bigr].
\]
The controller does \emph{not} know $H_{\mathrm{drift}}$ and observes only the
measurement-derived Lyapunov signal.

\medskip
\noindent\textbf{Sampled-data implementation (zero-order hold).}
Let $t_n=n\tau$ denote the sampling instants. At each $t_n$, the controller
obtains $V(t_n)$ from measurement outcomes, forms the finite difference
\[
\Delta V(t_n):=V(t_n)-V(t_{n-1}),
\]
and updates the control inputs according to the sign-based rule
\[
u_x(t_n)=-\kappa_x(t_n)\,\mathrm{sign}\!\bigl(\Delta V(t_n)\bigr),\qquad
u_y(t_n)=-\kappa_y(t_n)\,\mathrm{sign}\!\bigl(\Delta V(t_n)\bigr).
\]
The inputs are then held constant on the interval $[t_n,t_{n+1})$, i.e.,
\[
u_\alpha(t)=u_\alpha(t_n), \qquad t\in[t_n,t_{n+1}),\ \alpha\in\{x,y\}.
\]
The adaptive gains are updated on the same sampling grid:
\[
\kappa_\alpha(t_{n+1})
=\kappa_\alpha(t_n)+\alpha_\alpha\,|\Delta V(t_n)|,
\qquad \alpha\in\{x,y\}.
\]

This controller is model-free: it computes neither $\dot\rho$ nor any drift
estimate, and cannot cancel $H_{\mathrm{drift}}$.  
By Section~\ref{sec:iss}, the resulting behavior is ISS-like: convergence to a
drift-dependent neighborhood.

All assumptions of the finite-difference LaSalle principle are satisfied:
\begin{enumerate}
	\item $V(t_n)\in[0,1]$ is proper and bounded;
	\item $\{H_x,H_y\}$ ensure reachability of the target state;
	\item the sign rule and gain growth enforce finite-difference descent along the
	sampling instants whenever $\rho(t_n)\neq P_0$;
	\item unknown drift prevents exact convergence, but the ISS bound of
	Lemma~\ref{lemma:quantum-iss} guarantees convergence to a small
	disturbance-limited neighborhood.
\end{enumerate}

In simulation, $V(t_n)$ decreases after an oscillatory transient and settles at a
small nonzero value, in full agreement with the ISS theory and
Theorem~\ref{thm:no-asymptotic}.

Figure~\ref{fig:V_drift_doubleprobe} displays the Lyapunov observable.
The trajectory remains oscillatory throughout the evolution--due to the unknown
drift Hamiltonian--but the oscillations exhibit a gradually decreasing
amplitude.
This behavior is not a numerical artifact, but an inherent consequence of the
presence of unmodeled Hamiltonian drift.

In accordance with the ISS-type analysis, the observable settles into
a small, drift-limited residual level instead of converging to zero.
Such persistent but bounded oscillations are therefore expected and reflect
	practical (disturbance-limited) stabilization: while exact asymptotic
	convergence is precluded by the unknown drift, the closed-loop system is driven
	into a neighborhood of the target state whose size depends on the disturbance
	magnitude and admissible control amplitudes.

\begin{figure}[htb] 
	\centering
	\includegraphics[width=0.75\textwidth]{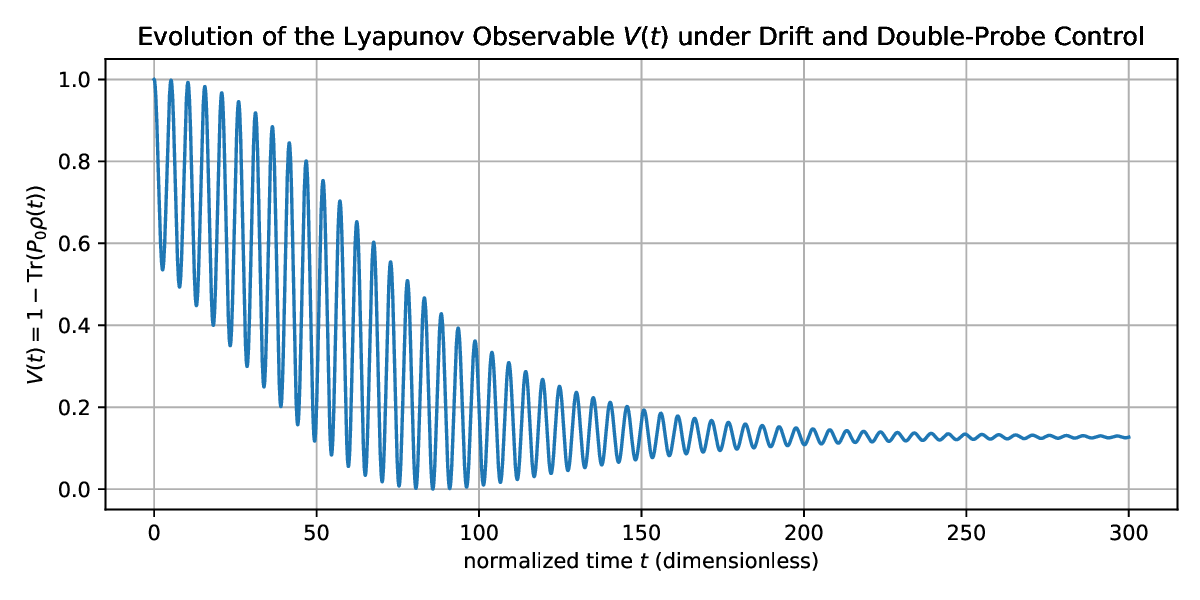}
	\caption{Evolution of the Lyapunov observable $V(t)$ under unknown drift and the
			double-probe model-free controller.
			The continuous-time signal exhibits persistent oscillations induced by the
			unknown drift and the zero-order-hold implementation of the control inputs.
			Lyapunov descent is enforced in the finite-difference sense at the sampling
			instants $t_n=t_0+n\tau$, rather than as a pointwise monotonic decrease of
			$V(t)$ in continuous time.
			As predicted by the ISS-type analysis, the oscillation amplitude gradually
			decreases and the trajectory approaches a disturbance-limited residual level.
			The time axis is expressed in normalized (dimensionless) units determined by
			the chosen scaling of the Hamiltonians and control amplitudes.
	}
			\label{fig:V_drift_doubleprobe}
\end{figure}

Figure~\ref{fig:controls_drift_doubleprobe} shows the control inputs.
During the initial transient ($t<50$), the $u_x$ channel exhibits
larger oscillation amplitudes than $u_y$, reflecting that the double-probe
estimator initially identifies a steeper descent direction along the $H_x$
control axis. As the Lyapunov observable decreases and the state moves closer
to its drift-limited equilibrium, the amplitudes of the two control channels
become comparable. The steady negative value of $u_x$ is not problematic:
its sign merely indicates the direction of the rotation generated by $H_x$
and does not carry any physical restriction or instability implication.
In particular, the control inputs are explicitly saturated in the simulations 
	to reflect physically admissible amplitude constraints. The observed behavior
	therefore demonstrates that the proposed feedback law does not rely on
	unbounded control amplitudes, but enforces Lyapunov descent whenever this is
	physically feasible.

More precisely, near the steady state the empirical Lyapunov gradient satisfies
\[
g_x \approx \frac{\partial V}{\partial \theta_x}, 
\qquad 
g_y \approx \frac{\partial V}{\partial \theta_y},
\]
where $\theta_k$ parametrizes infinitesimal rotations generated by $H_k$.
Because the controller applies
\[
u_x = -\lambda\, g_x, \qquad u_y = -\lambda\, g_y, \qquad (\lambda > 0),
\]
the sign of $u_x$ is determined by the sign of
$\partial V / \partial \theta_x$ evaluated at the drift-limited equilibrium.
Since this equilibrium does not coincide with the target state, the residual
Hamiltonian drift breaks the symmetry of the Lyapunov landscape and induces a
nonzero local gradient, yielding
\[
\frac{\partial V}{\partial \theta_x} > 0
\quad\Rightarrow\quad
u_x < 0.
\]
Thus, the observed negative steady-state value of $u_x$ simply reflects the
direction in which an infinitesimal rotation about $H_x$ would increase the
Lyapunov observable, and the controller compensates by applying a rotation in the
opposite direction.

\medskip
This behavior is fully consistent with the ISS-type analysis developed in
Section~\ref{sec:iss}, which predicts convergence to a disturbance-limited
neighborhood instead of exact asymptotic stabilization.

\begin{figure}[htb]
	\centering
	\includegraphics[width=0.75\textwidth]{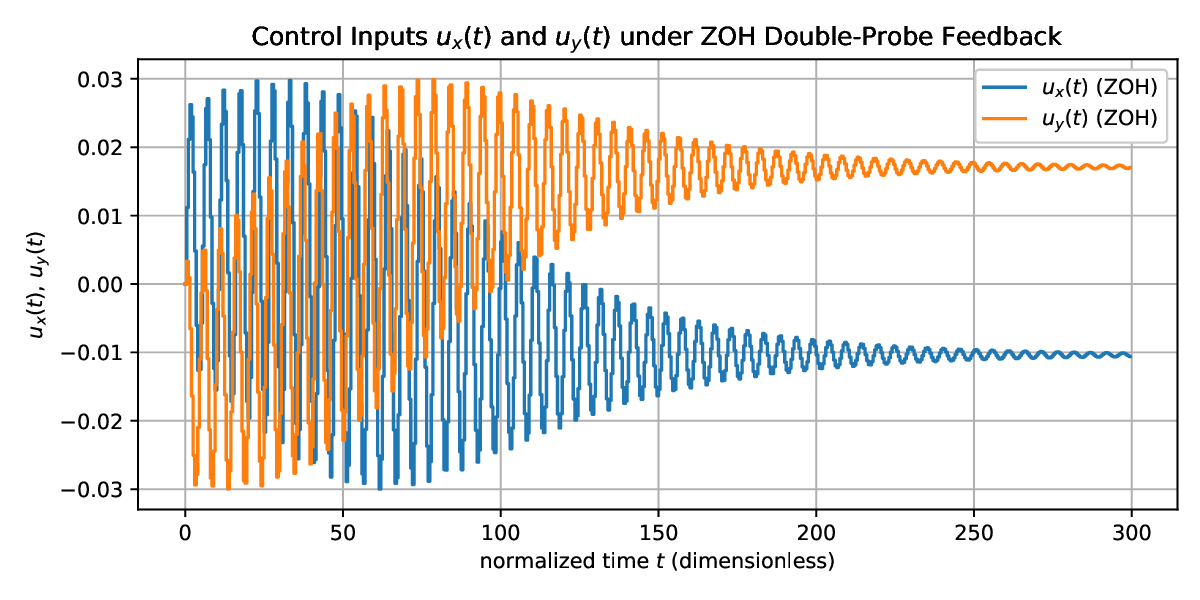}
	\caption{
		Control inputs $u_x(t)$ and $u_y(t)$ generated by the double-probe controller.
		Both channels remain oscillatory as they counteract the unknown drift.
		The control inputs are implemented under zero-order hold with sampling interval
			$\tau = 0.5$, and are therefore piecewise constant on each interval
			$[t_n,t_{n+1})$. While the overall oscillatory behavior is emphasized at the
			displayed scale, a zoomed view reveals the underlying stepwise (piecewise
			constant) structure induced by the sampled-data implementation.
			The time axis is the same normalized time scale as in
			Fig.~\ref{fig:V_drift_doubleprobe}.
	}
		\label{fig:controls_drift_doubleprobe}
\end{figure}

Figure~\ref{fig:bloch_drift_doubleprobe} displays the evolution of the Bloch
components $(x(t),y(t),z(t))$, where each coordinate is defined by
\[
x(t)=\Tr\!\bigl(\sigma_x \rho(t)\bigr),\qquad
y(t)=\Tr\!\bigl(\sigma_y \rho(t)\bigr),\qquad
z(t)=\Tr\!\bigl(\sigma_z \rho(t)\bigr),
\]
see, e.g., the standard Bloch-sphere representation of qubit states
in~\cite{nielsen_chuang_2010}.
These quantities are obtained from measurement statistics associated with the
Pauli operators and together form the Bloch vector associated with the qubit
state $\rho(t)$; geometrically, they determine the point representing $\rho(t)$
on the Bloch sphere.

Each coordinate exhibits persistent but gradually diminishing oscillations,
which are a hallmark of the underlying unknown drift Hamiltonian.
As the controller counteracts the drift using only finite-difference
information, the oscillation amplitudes decrease and the trajectory approaches a
drift-limited steady configuration in all three coordinates.
Thus, although the state is steered toward the vicinity of the north pole
$(0,0,1)$ corresponding to the target state $\ket{0}$, it does not converge
exactly to that point—consistent with the ISS-type limitation proved in
Section~\ref{sec:iss}.

\begin{figure}[htb]
	\centering
	\includegraphics[width=0.75\textwidth]{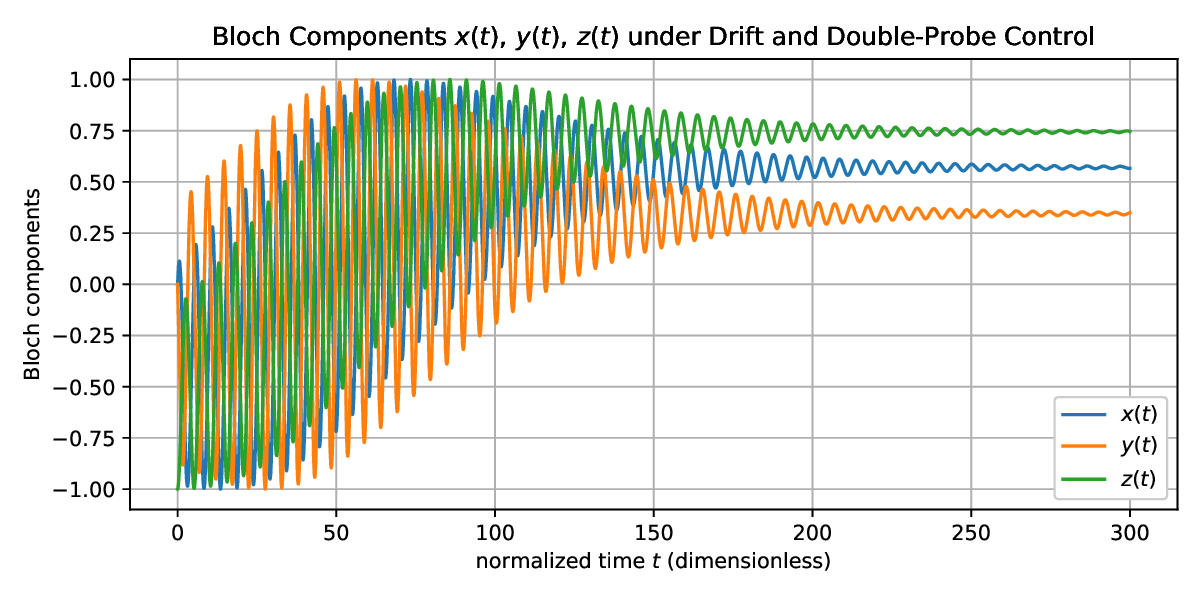}
	\caption{
		Bloch components $(x(t),y(t),z(t))$ under the double-probe controller in
		the presence of unknown drift.  
		All coordinates exhibit oscillations with decreasing amplitude and
		converge to a drift-limited equilibrium instead of the target pole.
		The time axis is the same normalized time scale as in
			Fig.~\ref{fig:V_drift_doubleprobe}.
	}
		\label{fig:bloch_drift_doubleprobe}
\end{figure}

To visualize the geometry of this behavior, 
Figure~\ref{fig:blochsphere_drift_doubleprobe} shows the corresponding
trajectory on the Bloch sphere.  
Starting from the south pole (orthogonal to the target), the state spirals
upward while the controller repeatedly corrects drift-induced deviations.
Because the drift cannot be cancelled without model knowledge, the trajectory
eventually settles on a stationary point located near
\[
(x_\infty, y_\infty, z_\infty) 
\approx (0.5711,\,0.3451,\,0.7448), 
\]
estimated from the last $500$ simulation steps.  
This geometric picture matches precisely the impossibility theorem: 
in the presence of an unknown nonzero drift, the target state is not an
equilibrium of the closed-loop dynamics, hence exact asymptotic stabilization
cannot occur.

\begin{figure}[htb] 
	\centering 
	\includegraphics[width=0.75\textwidth]{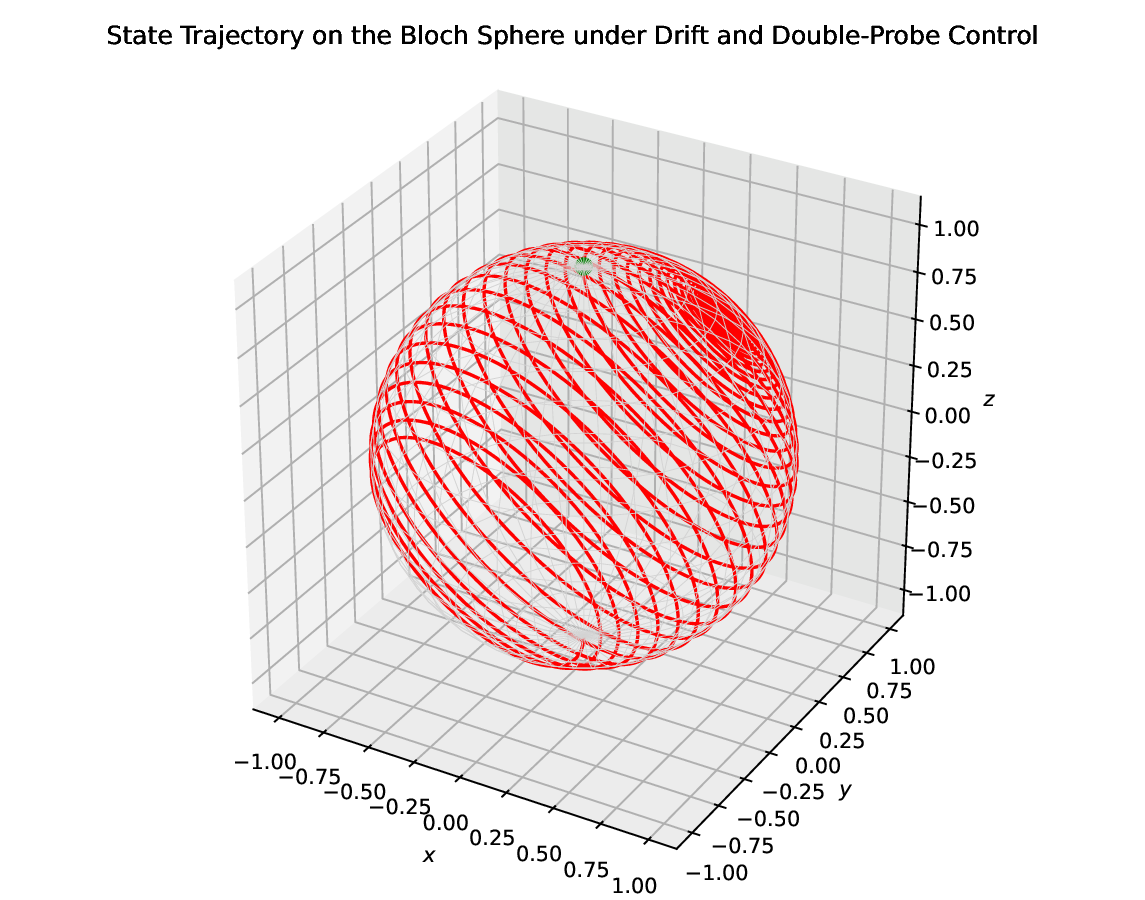}
	\caption{
		Bloch-sphere trajectory of the drifted qubit under the double-probe model-free controller.
		The state follows a decaying spiral and converges to a drift-limited equilibrium point instead of the ideal north pole, because the unknown Hamiltonian drift continuously rotates the state whenever
		$[\rho, H_{\mathrm{drift}}]\neq 0$.
		Only states commuting with the drift Hamiltonian can be fixed points; consequently, the controller achieves practical (ISS-type) stabilization to a small neighborhood whose size is determined by the drift strength.}
		\label{fig:blochsphere_drift_doubleprobe}
\end{figure}
   
\clearpage\newpage

\section{Conclusion}\label{sec:Conclusion}

We have developed a fully model-free framework for stabilizing quantum states
using only empirical evaluations of an adaptive Lyapunov observable.  
The controller requires no knowledge of the underlying generator~$\mathcal{F}$--
neither its Hamiltonian nor dissipative components--and relies exclusively on
finite-difference information obtained from measurement data.  
A simple sign-based feedback law, together with adaptive gain amplification,
enforces empirical Lyapunov descent without analytic derivatives or model
identification.

The central theoretical contribution is a finite-difference analogue of
LaSalle’s invariance principle.  
We show that sign-consistent feedback guarantees descent of the Lyapunov
observable whenever the system is away from the target, while adaptive gains
prevent stagnation at any positive Lyapunov level.  
Combined, these mechanisms ensure convergence of the Lyapunov observable along
the sampling instants and force its limit to zero in the drift-free case,
thereby establishing asymptotic stabilization.
When unknown drift or dissipation is present, the same structure yields an
ISS-type estimate, demonstrating practical stabilization to a disturbance-limited
neighbourhood of the target.

A single-qubit example illustrates the complete closed-loop mechanism.
The construction extends directly to arbitrary finite-dimensional quantum systems under the same uniform level-set descendability assumptions, without requiring any geometric conditions beyond physical realizability of the available controls and measurements.
Numerical simulations confirm the predicted behavior, including the fundamental ISS limitation in the presence of unknown disturbances.
Overall, the results provide a scalable and experimentally feasible paradigm for
quantum feedback based solely on finite-difference measurement data.  
The framework suggests several promising directions, including stochastic
extensions for weak measurements, stabilization of mixed states and subspaces,
performance-oriented adaptation schemes, and applications to multi-qubit and
multi-qudit architectures--pointing toward a broader theory of model-free quantum
control.

Finally, we emphasize that the proposed stabilization framework is intrinsically
	hybrid.  
	The quantum state $\rho(t)$ evolves in continuous time according to the underlying
	(open-loop or closed-loop) Schrödinger or Lindblad dynamics, while all feedback
	decisions--including Lyapunov evaluation, sign selection, and gain adaptation--
	are performed exclusively at discrete sampling instants $t_n=t_0+n\tau$ and
	applied under a zero-order hold assumption on the intervals $[t_n,t_{n+1})$.  
	
	All stability guarantees in this work are therefore formulated with respect to the
	sampled Lyapunov sequence $\{V(t_n)\}_{n\ge 0}$.
	This explicit separation resolves the apparent tension between continuous-time
	quantum evolution and measurement-driven feedback, and places the analysis firmly
	within a sampled-data control paradigm compatible with realistic experimental
	implementations.

\section*{Disclosure statement}
The authors declare that they have no financial or personal conflicts of interest that could have influenced the work reported in this manuscript.

\section*{Data availability statement}
No new data were created or analysed in this study. Therefore, data sharing is not applicable to this article.



\begin{thebibliography}{26}
		\providecommand{\natexlab}[1]{#1}
		\providecommand{\url}[1]{\texttt{#1}}
		\expandafter\ifx\csname urlstyle\endcsname\relax
		\providecommand{\doi}[1]{doi: #1}\else
		\providecommand{\doi}{doi: \begingroup \urlstyle{rm}\Url}\fi
		
		\bibitem[Wiseman and Milburn(2011)]{wiseman_milburn_2010}
		H.~M. Wiseman and G.~J. Milburn.
		\newblock \emph{Quantum Measurement and Control}.
		\newblock Cambridge University Press, 2011.
		\newblock URL \url{https://doi.org/10.1017/CBO9780511813948}.
		
		\bibitem[Weidner et~al.(2025)Weidner, Reed, Monroe, Sheller, O’Neil, Maas,
		Jonckheere, Langbein, and Schirmer]{weidner2025robust}
		Carrie~Ann Weidner, Emily~A. Reed, Jonathan Monroe, Benjamin Sheller, Sean
		O’Neil, Eliav Maas, Edmond~A. Jonckheere, Frank~C. Langbein, and Sophie
		Schirmer.
		\newblock Robust quantum control in closed and open systems: Theory and
		practice.
		\newblock \emph{Automatica}, 172:\penalty0 111987, 2025.
		\newblock \doi{10.1016/j.automatica.2024.111987}.
		\newblock URL \url{https://doi.org/10.1016/j.automatica.2024.111987}.
		
		\bibitem[Dong and Petersen(2010)]{dong_petersen_2010}
		Daoyi Dong and Ian~R. Petersen.
		\newblock Quantum control theory and applications: a survey.
		\newblock \emph{IET Control Theory \& Applications}, 4\penalty0 (12):\penalty0
		2651--2671, 2010.
		\newblock URL \url{https://doi.org/10.1049/iet-cta.2009.0508}.
		
		\bibitem[Altafini and Ticozzi(2012)]{altafini_ticozzi_2012}
		Claudio Altafini and Francesco Ticozzi.
		\newblock Modeling and control of quantum systems: an introduction.
		\newblock \emph{IEEE Transactions on Automatic Control}, 57\penalty0
		(8):\penalty0 1898--1917, 2012.
		\newblock URL \url{https://doi.org/10.1109/TAC.2012.2195830}.
		
		\bibitem[Ticozzi et~al.(2010)Ticozzi, Schirmer, and Wang]{ticozzi_viola_2008}
		Francesco Ticozzi, Sophie~G. Schirmer, and Xiaoting Wang.
		\newblock Stabilizing quantum states by constructive design of open quantum
		dynamics.
		\newblock \emph{IEEE Transactions on Automatic Control}, 55\penalty0
		(12):\penalty0 2901--2905, 2010.
		\newblock \doi{10.1109/TAC.2010.2079532}.
		\newblock URL
		\url{https://ieeexplore.ieee.org/document/5585722/similar#similar}.
		
		\bibitem[Belavkin(1983)]{belavkin_1983}
		Viacheslav Belavkin.
		\newblock Towards the theory of control in observable quantum systems.
		\newblock \emph{Automatica and Remote Control}, 44:\penalty0 178--188, 1983.
		\newblock URL \url{https://doi.org/10.48550/arXiv.quant-ph/0408003}.
		
		\bibitem[Wiseman and Milburn(1993)]{wiseman_milburn_1993}
		H.~M. Wiseman and G.~J. Milburn.
		\newblock Quantum theory of optical feedback via homodyne detection.
		\newblock \emph{Physical Review Letters}, 70:\penalty0 548--551, 1993.
		\newblock URL \url{https://doi.org/10.1103/PhysRevLett.70.548}.
		
		\bibitem[Bouten et~al.(2007)Bouten, van Handel, and James]{bouten_2007}
		Luc Bouten, Ramon van Handel, and Matthew~R. James.
		\newblock An introduction to quantum filtering.
		\newblock \emph{SIAM Journal on Control and Optimization}, 46\penalty0
		(6):\penalty0 2199--2241, 2007.
		\newblock URL \url{https://doi.org/10.1137/060651239}.
		
		\bibitem[Khalil(2002)]{khalil_2002}
		Hassan~K. Khalil.
		\newblock \emph{Nonlinear Systems}.
		\newblock Prentice Hall, Upper Saddle River, {N.J.}, 3 edition, 2002.
		\newblock ISBN 9780130673893.
		
		\bibitem[Sontag(1989)]{sontag_1989}
		Eduardo~D. Sontag.
		\newblock Smooth stabilization implies coprime factorization.
		\newblock \emph{IEEE Transactions on Automatic Control}, 34\penalty0
		(4):\penalty0 435--443, 1989.
		\newblock URL \url{https://ieeexplore.ieee.org/stamp/stamp.jsp?arnumber=28018}.
		
		\bibitem[Jiang et~al.(1996)Jiang, Mareels, and Wang]{jiang_mareels_wang_1996}
		Zhong-Ping Jiang, Iven~M.Y. Mareels, and Yuan Wang.
		\newblock A {L}yapunov formulation of the nonlinear small-gain theorem for
		interconnected {ISS} systems.
		\newblock \emph{Automatica}, 32\penalty0 (8):\penalty0 1211--1215, 1996.
		\newblock URL \url{https://doi.org/10.1016/0005-1098(96)00051-9}.
		
		\bibitem[Ticozzi et~al.(2012)Ticozzi, Lucchese, Cappellaro, and
		Viola]{ticozzi_viola_2012}
		Francesco Ticozzi, Riccardo Lucchese, Paola Cappellaro, and Lorenza Viola.
		\newblock Hamiltonian control of quantum dynamical semigroups: Stabilization
		and convergence speed.
		\newblock \emph{IEEE Transactions on Automatic Control}, 57\penalty0
		(8):\penalty0 1931--1944, 2012.
		\newblock URL \url{https://ieeexplore.ieee.org/document/6189050}.
		
		\bibitem[Emzir et~al.(2022)Emzir, Woolley, and Petersen]{Emzir2022}
		Muhammad~Fuady Emzir, Matthew~J. Woolley, and Ian~R. Petersen.
		\newblock Stability analysis of quantum systems: {A} {L}yapunov criterion and
		an invariance principle.
		\newblock \emph{Automatica}, 146:\penalty0 110660, 2022.
		\newblock URL \url{https://doi.org/10.1016/j.automatica.2022.110660}.
		
		\bibitem[Wu et~al.(2025)Wu, Xue, Ma, Kuang, Dong, and
		Petersen]{wu_switching_qubit_2025}
		Guangpu Wu, Shibei Xue, Shan Ma, Sen Kuang, Daoyi Dong, and Ian~R. Petersen.
		\newblock Arbitrary state transition of open qubit system based on switching
		control.
		\newblock \emph{Automatica}, 179:\penalty0 112424, 2025.
		\newblock \doi{10.1016/j.automatica.2025.112424}.
		\newblock URL
		\url{https://www.sciencedirect.com/science/article/pii/S0005109825003188?via%3Dihub}.
		
		\bibitem[Lindblad(1976)]{lindblad_1976}
		Göran Lindblad.
		\newblock On the generators of quantum dynamical semigroups.
		\newblock \emph{Communications in Mathematical Physics}, 48:\penalty0 119--130,
		1976.
		\newblock URL \url{https://doi.org/10.1007/BF01608499}.
		
		\bibitem[Gorini et~al.(1976)Gorini, Kossakowski, and
		Sudarshan]{gorini_kossakowski_sudarshan_1976}
		Vittorio Gorini, Andrzej Kossakowski, and E.~C.~G. Sudarshan.
		\newblock Completely positive dynamical semigroups of {N}-level systems.
		\newblock \emph{Journal of Mathematical Physics}, 17\penalty0 (5):\penalty0
		821--825, 1976.
		\newblock URL \url{https://doi.org/10.1063/1.522979}.
		
		\bibitem[Pan et~al.(2014)Pan, Amini, Miao, Gough, Ugrinovskii, and
		James]{pan_amini_miao_gough_ugrinovskii_james_2014}
		Yu~Pan, Hadis Amini, Zibo Miao, John Gough, Valery Ugrinovskii, and Matthew~R.
		James.
		\newblock Heisenberg picture approach to the stability of quantum {M}arkov
		systems.
		\newblock \emph{Journal of Mathematical Physics}, 55\penalty0 (6):\penalty0
		062701, 2014.
		\newblock \doi{10.1063/1.4884300}.
		\newblock URL \url{https://doi.org/10.1063/1.4884300}.
		
		\bibitem[Song et~al.(2025)Song, Liu, Dong, and
		Yonezawa]{song_drl_feedback_2025}
		Chunxiang Song, Yanan Liu, Daoyi Dong, and Hidehiro Yonezawa.
		\newblock Fast state stabilization using deep reinforcement learning for
		measurement-based quantum feedback control.
		\newblock \emph{IEEE Transactions on Quantum Engineering}, 6, 2025.
		\newblock \doi{10.1109/TQE.2025.3606123}.
		\newblock URL
		\url{https://ieeexplore.ieee.org/stamp/stamp.jsp?arnumber=11150735}.
		
		\bibitem[Burgarth and Yuasa(2012)]{burgarth_yuasa_prl_2012}
		Daniel Burgarth and Kazuya Yuasa.
		\newblock Quantum system identification.
		\newblock \emph{Physical Review Letters}, 108:\penalty0 080502, 2012.
		\newblock \doi{10.1103/PhysRevLett.108.080502}.
		\newblock URL
		\url{https://journals.aps.org/prl/abstract/10.1103/PhysRevLett.108.080502}.
		
		\bibitem[Petersen et~al.(2012)Petersen, Ugrinovskii, and
		James]{petersen_ugrinovskii_james_2012}
		Ian~R. Petersen, Valery Ugrinovskii, and Matthew~R. James.
		\newblock Robust stability of uncertain linear quantum systems.
		\newblock \emph{Philosophical Transactions of the Royal Society A},
		370\penalty0 (1979):\penalty0 5354--5363, 2012.
		\newblock \doi{10.1098/rsta.2011.0527}.
				
		\bibitem[Zhang and Sarovar(2014)]{zhang_sarovar_2014}
		Jun Zhang and Mohan Sarovar.
		\newblock Quantum hamiltonian identification from measurement time traces.
		\newblock \emph{Physical Review Letters}, 113\penalty0 (8):\penalty0 080401,
		2014.
		\newblock \doi{10.1103/PhysRevLett.113.080401}.
		\newblock URL \url{https://doi.org/10.1103/PhysRevLett.113.080401}.
		
		\bibitem[Bukov et~al.(2018)Bukov, Day, Sels, Weinberg, Polkovnikov, and
		Mehta]{bukov_rl_2018}
		Marin Bukov, Alexandre G.~R. Day, Dries Sels, Phillip Weinberg, Anatoli
		Polkovnikov, and Pankaj Mehta.
		\newblock Reinforcement learning in different phases of quantum control.
		\newblock \emph{Physical Review X}, 8\penalty0 (3):\penalty0 031086, 2018.
		\newblock \doi{10.1103/PhysRevX.8.031086}.
		\newblock URL \url{https://doi.org/10.1103/PhysRevX.8.031086}.
		
		\bibitem[Niu et~al.(2019)Niu, Boixo, Smelyanskiy, and Neven]{niu_rl_2019}
		Murphy~Yuezhen Niu, Sergio Boixo, Vadim~N. Smelyanskiy, and Hartmut Neven.
		\newblock Universal quantum control through deep reinforcement learning.
		\newblock \emph{npj Quantum Information}, 5:\penalty0 33, 2019.
		\newblock \doi{10.1038/s41534-019-0141-3}.
		\newblock URL \url{https://www.nature.com/articles/s41534-019-0141-3}.
		
		\bibitem[Clerk et~al.(2010)Clerk, Devoret, Girvin, Marquardt, and
		Schoelkopf]{clerk_2010}
		A.~A. Clerk, M.~H. Devoret, S.~M. Girvin, Florian Marquardt, and R.~J.
		Schoelkopf.
		\newblock Introduction to quantum noise, measurement, and amplification.
		\newblock \emph{Reviews of Modern Physics}, 82\penalty0 (2):\penalty0
		1155--1208, 2010.
		\newblock \doi{10.1103/RevModPhys.82.1155}.
		\newblock URL \url{https://doi.org/10.1103/RevModPhys.82.1155}.
		
		\bibitem[D'Alessandro(2007)]{d_alessandro_2007}
		Domenico D'Alessandro.
		\newblock \emph{Introduction to Quantum Control and Dynamics}.
		\newblock Chapman and Hall/CRC Press, 2007.
		\newblock ISBN 9781584888833.
		\newblock URL \url{https://doi.org/10.1201/9781584888833}.
		
		\bibitem[Nielsen and Chuang(2010)]{nielsen_chuang_2010}
		Michael~A. Nielsen and Isaac~L. Chuang.
		\newblock \emph{Quantum Computation and Quantum Information: 10th Anniversary
			Edition}.
		\newblock Cambridge University Press, 2010.
		\newblock ISBN 9781107002173.
		
	\end{thebibliography}
\end{document}